\newif\ifFULL
\FULLtrue

\ifFULL
\documentclass[11pt,english]{article}
\else
\documentclass[10pt, conference, compsocconf]{IEEEtran}
\usepackage[english]{babel}
\fi

\usepackage[latin9]{inputenc}
\usepackage{geometry}
\usepackage{amsthm}
\usepackage{amsmath}
\usepackage{amssymb}

\makeatletter
\theoremstyle{plain}
\newtheorem{thm}{\protect\theoremname}[section]
\newtheorem{lem}[thm]{Lemma}
\newtheorem{cor}[thm]{Corollary}
\newtheorem{conj}[thm]{Conjecture}

\newenvironment{psketch}{%
  \proof}{\endproof}

\newtheorem{oq}{Open Question}

  \theoremstyle{definition}
  \newtheorem{defn}[thm]{\protect\definitionname}
\newtheorem{remark}[thm]{Remark}

\usepackage{complexity,color,bm}
\usepackage{colortbl}
\newcommand{\ignore}[1]{}%

\usepackage{nicefrac}

\theoremstyle{plain}

\DeclareMathOperator*{\Lreg}{L}

\makeatother

\usepackage{babel}
  \providecommand{\definitionname}{Definition}
\providecommand{\theoremname}{Theorem}

\newcounter{note}

\ifFULL

\title{Distributed PCP Theorems for Hardness of Approximation in P}

\author{%
\normalsize
Amir Abboud \thanks{IBM Almaden Research Center,
amir.abboud@ibm.com. Research done while at Stanford University and supported by the grants of Virginia Vassilevska Williams: NSF Grants CCF-1417238, CCF-1528078 and CCF-1514339, and BSF Grant BSF:2012338.}
\and
\normalsize
Aviad Rubinstein \thanks{UC Berkeley,
aviad@eecs.berkeley.edu. This research was supported by Microsoft Research PhD Fellowship. It was also supported in part by NSF grant CCF-1408635 and by Templeton Foundation grant 3966.}
\and
\normalsize
Ryan Williams \thanks{MIT, rrw@mit.edu. Supported by an NSF CAREER grant (CCF-1741615).}
}

\date{}

\else

\title{Distributed PCP Theorems for Hardness of Approximation in P \\ {\Large (Extended Abstract)}}

\author{\IEEEauthorblockN{Amir Abboud}
\IEEEauthorblockA{Stanford University\\
Computer Science Department\\
Palo Alto, CA, USA\\
abboud@cs.stanford.edu}
\and
\IEEEauthorblockN{Aviad Rubinstein}
\IEEEauthorblockA{UC Berkeley\\
EECS\\
Berkeley, CA, USA\\
aviad@eecs.berkeley.edu}
\and
\IEEEauthorblockN{Ryan Williams}
\IEEEauthorblockA{MIT\\
CSAIL and EECS\\
Cambridge, MA, USA\\
rrw@mit.edu}
}

\fi

\begin{document}

\maketitle

\begin{abstract}
We present a new \emph{distributed} model of probabilistically checkable proofs (PCP). A satisfying assignment $x \in \{0,1\}^n$ to a CNF formula $\varphi$ is shared between two parties, where Alice knows $x_1, \dots, x_{n/2}$, Bob knows $x_{n/2+1},\dots,x_n$, and both parties know $\varphi$. The goal is to have Alice and Bob jointly write a PCP that $x$ satisfies $\varphi$, while exchanging little or no information. 
Unfortunately, this model as-is does not allow for nontrivial query complexity. Instead, we focus on a \emph{non-deterministic} variant, where the players are helped by Merlin, a third party who knows all of $x$.

Using our framework, we obtain, for the first time, PCP-like reductions from the Strong Exponential Time Hypothesis (SETH) to approximation problems in \P.
In particular, under SETH we show that 
there are no truly-subquadratic approximation algorithms for 
{\bf Bichromatic Maximum Inner Product} over $\{0,1\}$-vectors,
{\bf Bichromatic LCS Closest Pair} over permutations, 
{\bf Approximate Regular Expression Matching}, and 
{\bf Diameter in Product Metric}.
All our inapproximability factors are nearly-tight. In particular, for the first two problems we obtain nearly-polynomial factors of $2^{(\log n)^{1-o(1)}}$;
only $(1+o(1))$-factor lower bounds (under SETH) were known before.
\end{abstract}

\ifFULL

\clearpage
\thispagestyle{empty}


\setcounter{page}{1}

\else


%
\IEEEpeerreviewmaketitle

\fi

\section{Introduction}\label{sec:intro}
\def \eps{\varepsilon}

Fine-Grained Complexity classifies the time complexity of fundamental problems under popular conjectures, the most productive of which has been the Strong Exponential Time Hypothesis%
\footnote{SETH is a pessimistic version of \P $\neq$ \NP, stating that for every $\eps > 0$ there is a $k$ such that $k$-SAT cannot be solved in $O((2-\eps)^n)$ time.} (SETH). The list of ``SETH-Hard" problems is long, including central problems in pattern matching and bioinformatics \cite{AVW14,BI15,BI16}, graph algorithms \cite{RV13,GIKW17}, dynamic data structures \cite{AV14}, parameterized complexity and exact algorithms \cite{PW10-faster_SAT,LMS11,Cygan+16}, computational geometry \cite{Bring14}, time-series analysis \cite{ABV15a,BK15}, and even economics \cite{MPS16} (a longer list can be found in \cite{Vass15}).

For most problems in the above references, there are natural and meaningful approximate versions, and for most of them the time complexity is wide open (a notable exception is~\cite{RV13}).
Perhaps the most important and challenging open question in the field of Fine-Grained Complexity is whether a framework for \emph{hardness of approximation in \P} is possible.
To appreciate the gaps in our knowledge regarding inapproximability, 
consider the following fundamental problem from the realms of {similarity search} and statistics, of finding the most \emph{correlated} pair in a dataset.

\begin{defn}[The {\sc Bichromatic%
\footnote{See the end of this section for a discussion of ``bichromatic'' vs ``monochromatic'' closest pair problems.} 
Max Inner Product} Problem ({\sc Max-IP})]
Given two sets $A,B$, each of $N$ binary vectors in $\{0,1\}^d$, return a pair $(a,b) \in A \times B$ that maximizes the inner product $a\cdot b$.
\end{defn}

Thinking of the vectors as subsets of $[d]$, this {\sc Max-IP} problem asks to find the pair with largest overlap, a natural similarity measure. 
A na\"{i}ve algorithm solves the problem in $O(N^2d)$ time, 
and one of the most-cited fine-grained results is a SETH lower bound for this problem.
Assuming SETH, we cannot solve {\sc Max-IP} (exactly) in $N^{2-\eps}\cdot 2^{o(d)}$ time, for any $\eps>0$ \cite{Wil05}.

This lower bound is hardly pleasing when one of the most vibrant areas of Algorithms\footnote{In SODA'17, two entire sessions were dedicated to algorithms for similarity search.} 
is concerned with designing \emph{approximate} but \emph{near-linear} time solutions for  such similarity search problems. For example, the original motivation of the 
celebrated MinHash algorithm was to solve the indexing version of this problem~\cite{Broder97,Broder+97}, and one of the first implementations was at the core of the AltaVista search engine. 
The problem has important applications all across Computer Science, most notably in Machine Learning, databases, and information retrieval, e.g. \cite{IM98,AI06,RR07,RG12,SL14,AINR14,AILRS15,AR15,NS15,SL15,Valiant15,AW15,KarKK16,pods16,TG16,CP16,Chris17}.

{\sc Max-IP} seems to be more challenging than closely related problems where similarity is defined as small Euclidean distance rather than large inner product. 
For the latter, we can get near-linear $O(N^{1+\eps})$ time algorithms, for all $\eps>0$, at the cost of some constant $f(\eps)$ error that depends on $\eps$ \cite{IM98,AI06,AINR14,AR15}.
In contrast, for {\sc Max-IP}, even for a moderately subquadratic running time of $O(N^{2-\eps})$, all known algorithms suffer from \emph{polynomial} $N^{g(\eps)}$ approximation factors.

Meanwhile, the SETH lower bound for {\sc Max-IP} was only slightly improved by Ahle, Pagh, Razenshteyn, and Silvestri \cite{pods16} to rule out $1+o(1)$ approximations, leaving a huge gap between the not-even-$1.001$ lower bound and the polynomial upper bound.

\begin{oq}
\label{oq1}
{Is there an $O(N^{1+\eps})$-time algorithm for computing an $f(\eps)$-approximation to {\sc Bichromatic Max Inner Product} over binary vectors?}
\end{oq}

This is just one of the many significant open questions that highlight our inability to prove hardness of approximation in \P, and pour cold water on the excitement 
from the successes of Fine-Grained Complexity. 
It is natural to try to adapt tools from the \NP-Hardness-of-approximation framework (namely, the celebrated \emph{PCP Theorem}) to \P. Unfortunately, when starting from SETH, almost everything in the existing theory of PCPs breaks down. Whether PCP-like theorems for Fine-Grained Complexity are possible, and what they could look like, are fascinating open questions. 

Our main result is the first SETH-based PCP-like theorem, from which several strong hardness of approximation in \P~results follow.
We identify a canonical problem 
that is hard to approximate, and further gadget-reductions allow us to prove  
SETH-based inapproximability results for basic problems such as Subset Queries, Closest Pair under the Longest Common Subsequence similarity measure, and Furthest Pair (Diameter) in product metrics.
In particular, assuming SETH, we negatively resolve Open Question~\ref{oq1} in a very strong way, proving an almost tight lower bound for {\sc Max-IP}.

\subsection{PCP-like Theorems for Fine-Grained Complexity}

The following meta-structure is common to most SETH-based
reductions: given a CNF $\varphi$, construct $N=O\left(2^{\frac{n}{2}}\right)$ gadgets, one for each assignment to the first/last $n/2$ variables, and embed those gadgets into some problem $A$. The embedding is designed so that if $A$ can be solved in $O\left(N^{2-\varepsilon}\right)=O\left(2^{\left(1-\frac{\varepsilon}{2}\right)n}\right)$
time, a satisfying assignment for $\varphi$ can be efficiently recovered
from the solution, contradicting SETH.

The most obvious barrier to proving fine-grained hardness of approximation
is the lack of an appropriate PCP theorem. Given a 3-SAT formula $\varphi$,
testing that an assignment $x\in\left\{ 0,1\right\} ^{n}$ satisfies
$\varphi$ requires reading all $n$ bits of $x$. The PCP Theorem~\cite{AS98-PCP, ALMSS98-PCP}, shows how to transform $x\in\left\{ 0,1\right\} ^{n}$
into a PCP ({\em probabilistically checkable proof}) $\pi=\pi\left(\varphi,x\right)$,
which can be tested by a probabilistic verifier who only reads a few bits from $\pi$. This is the starting point for almost all proofs of
\NP-hardness of approximation. The main obstacle in using PCPs for
fine-grained hardness of approximation is that all known PCPs incur
a blowup in the size proof: $\pi\left(\varphi,x\right)$ requires
$n'\gg n$ bits. The most efficient known PCP, due to Dinur~\cite{Dinur07-PCP},
incurs a polylogarithmic blowup ($n'=n\cdot\polylog(n)$), and obtaining
a PCP with a constant blowup is a major open problem (e.g.~\cite{BKKMS16-linear_PCP, Dinur16-gap_ETH}). However, note that even if we had a fantastic PCP with only $n'=10n$, a reduction of size $N'=2^{\frac{n'}{2}}=2^{5n}$ does not imply any hardness at all.
Our goal is to overcome this barrier:

\begin{oq}
\label{oq2}
{Is there a PCP-like theorem for fine-grained complexity?}
\end{oq}

\subsection*{Distributed PCPs}

Our starting point is that of error-correcting codes, a fundamental building block of PCPs. Suppose that Alice and Bob want to encode a message $m=\left(\alpha;\beta\right)\in\left\{ 0,1\right\} ^{n}$
in a distributed fashion. Neither Alice nor Bob knows the entire message:
Alice knows the first half ($\alpha\in\left\{ 0,1\right\} ^{\frac{n}{2}}$),
and Bob knows the second half ($\beta\in\left\{ 0,1\right\} ^{\frac{n}{2}}$).
Alice can locally compute an encoding $E'\left(\alpha\right)$
of her half, and Bob locally computes an encoding $E'\left(\beta\right)$
of his. 
Then the concatenation of Alice's and Bob's strings, 
$E\left(m\right)=\left(E'\left(\alpha\right);E'\left(\beta\right)\right)$,
is an error-correcting encoding of $m$. 

Now let us return to {\em distributed PCPs}. Alice and Bob share
a $k$-SAT formula%
 \footnote{In the formulation of SETH, $k$ is a ``sufficiently large constant''%
\ifFULL 
 (see Section~\ref{sec:prelim} for definition).
\else
.
\fi
However, for the purposes of our discussion here it suffices to think of $k=3$.}
 $\varphi$. 
Alice has an assignment $\alpha\in\left\{ 0,1\right\} ^{\frac{n}{2}}$
to the first half of the variables, and Bob has an assignment $\beta\in\left\{ 0,1\right\} ^{\frac{n}{2}}$
to the second half. We want a protocol where Alice locally computes
a string $\pi'\left(\alpha\right)\in\left\{ 0,1\right\} ^{n'}$,
Bob locally computes $\pi'\left(\beta\right)\in\left\{ 0,1\right\} ^{n'}$,
and together $\pi\left(\alpha;\beta\right)=\left(\pi'\left(\alpha\right),\pi'\left(\beta\right)\right)$
is a valid probabilistically checkable proof that $x=\left(\alpha,\beta\right)$
satisfies $\varphi$. That is, a probabilistic verifier can read a
constant number of bits from $\left(\pi'\left(\alpha\right),\pi'\left(\beta\right)\right)$
and decide (with success probability at least $2/3$) whether $\left(\alpha,\beta\right)$
satisfies $\varphi$.

It is significant to note that {\bf if distributed PCPs can be constructed, then very strong reductions for fine-grained hardness of approximation follow}, completely overcoming the barrier for fine-grained PCPs outlined above. The reason is that
we can still construct $N=O\left(2^{\frac{n}{2}}\right)$ gadgets, one for
each half assignment $\alpha,\beta\in\left\{ 0,1\right\} ^{\frac{n}{2}}$,
where the gadget for $\alpha$ also encodes $\pi'\left(\alpha\right)$.
The blowup of the PCP only affects the size of each gadget, which
is negligible compared to the number of gadgets. In fact, this technique
would be so powerful, that we could reduce SETH to problems like approximate
$\ell_{2}$-Nearest Neighbor, where the existing sub-quadratic approximation
algorithms (e.g. \cite{AR15}) would falsify SETH!

Alas, distributed PCPs are {\em unconditionally impossible} (even for $2$-SAT)
by a simple reduction from Set Disjointness:
\begin{thm}
[Reingold~\cite{Reingold17-communication}; informal]
Distributed PCPs are impossible.\end{thm}
\begin{psketch} Consider the 2-SAT formula 
\ifFULL
$$\varphi\triangleq\bigwedge_{i=1}^{n/2}\left(\neg\alpha_{i}\vee\neg\beta_{i}\right).$$
\else
$\varphi\triangleq\bigwedge_{i=1}^{n/2}\left(\neg\alpha_{i}\vee\neg\beta_{i}\right)$.
\fi
This $\varphi$ is satisfied by assignment $\left(\alpha;\beta\right)$
iff the vectors $\alpha,\beta\in\left\{ 0,1\right\} ^{\frac{n}{2}}$ are disjoint. If a PCP verifier can decide whether $\left(\alpha;\beta\right)$
satisfies $\varphi$ by a constant number of queries to $\left(\pi'\left(\alpha\right),\pi'\left(\beta\right)\right)$,
then Alice and Bob can simulate the PCP verifier to decide whether
their vectors are disjoint, while communicating only a constant number
of bits (the values read by the PCP verifier). This contradicts the randomized communication complexity lower bounds of $\Omega(n)$ for set disjointness~\cite{KS92-communication, Razborov92-disjointness, BJKS04-communication}.
\end{psketch}

Note that the proof shows that even distributed PCPs with $o(n)$ queries are impossible.

\subsubsection*{Distributed and non-deterministic PCPs}

As noted above, set disjointness is very hard for randomized communication, and hard even for non-deterministic communication~\cite{KKN95-disjointness}. But Aaronson and Wigderson~\cite{AW09-algebrization} showed that set disjointness does have $\tilde{O}\left(\sqrt{n}\right)$ Merlin-Arthur ($\MA$) communication complexity.
In particular, they construct a simple protocol where the standard
Bob and an untrusted Merlin (who can see both sets of Alice and Bob)  
each send Alice a message of length $\tilde{O}\left(\sqrt{n}\right)$. If the sets are disjoint, Merlin can convince Alice to accept; if they are not, Alice will reject
with high probability regardless of Merlin's message. 

Our second main insight in this paper is this: for problems where
the reduction from SETH allows for an efficient OR gadget, we can
enumerate over all possible messages from Merlin and Bob%
\footnote{In fact, enumerating over Merlin's possible messages turns out to be easy to implement in the reductions; the main bottleneck is the communication with Bob.}. Thus we incur only a subexponential blowup%
\footnote{Subexponential in $n$ (the number of $k$-SAT variables), which implies sub{\em polynomial} in $N \approx 2^{n/2}$.} 
in the reduction size, while overcoming the communication barrier. Indeed, the construction in our PCP-like
theorem \ifFULL
(Theorem \ref{thm:PCP})
\fi
 can be interpreted as implementing a variant of
Aaronson and Wigderson's $\MA$ communication protocol. The resulting
PCP construction is {\em distributed} (in the sense described above) and {\em non-deterministic} (in the sense that Alice receives sublinear advice from Merlin).

It can be instructive to view our distributed PCP model as a 4-party (computationally-efficient) communication problem. Merlin wants to convince Alice, Bob, and Veronica (the verifier) that Alice and Bob jointly hold a satisfying assignment to a publicly-known formula.
Merlin sees everything except the outcome of random coin tosses, but he can only send $o(n)$ bits to only Alice. Alice and Bob each know half of the (allegedly) satisfying assignment, and each of them must (deterministically) send a (possibly longer) message to Veronica. Finally, Veronica tosses coins and is restricted to reading only $o(n)$ bits from Alice's and Bob's messages, after which she must output Accept/Reject.

Patrascu and Williams~\cite{PW10-faster_SAT} asked whether it
is possible to use Aaronson and Wigderson's $\MA$ protocol for Set Disjointness to obtain better algorithms for satisfiability. 
Taking an optimistic twist, our results in this paper may suggest this is indeed possible: if any of several important and simple problems admit efficient approximation
algorithms, then faster algorithms for (exact) satisfiability may be obtained
via Aaronson and Wigderson's $\MA$ protocol.

\subsection{Our results}

\ifFULL
Our distributed and non-deterministic PCP theorem is formalized as Theorem~\ref{thm:PCP}.
Since our main interest is proving hardness-of-approximation results, in Section~\ref{sec:main} we abstract the prover-verifier formulation by reducing our PCP to an Orthogonal-Vectors-like problem which we call 
{\sc PCP-Vectors}.
The hardness of {\sc PCP-Vectors} is formalized as Theorem~\ref{thm:main}. {\sc PCP-Vectors} turns out to be an excellent starting point for many results, yielding easy reductions for fundamental problems and giving essentially tight inapproximability bounds. Let us exhibit what we think are the most interesting ones.
\else
Our distributed and non-deterministic PCP theorem is formalized and proved in the full version.
Since our main interest is proving hardness-of-approximation results, we abstract the prover-verifier formulation by reducing our PCP to an Orthogonal-Vectors-like problem which we call {\sc PCP-Vectors} (see below).
{\sc PCP-Vectors} turns out to be an excellent starting point for many results, yielding easy reductions for fundamental problems and giving essentially tight inapproximability bounds. We begin with the description of {\sc PCP-Vectors}, and then exhibit what we think are the most interesting applications.

\paragraph{PCP-Vectors}
We introduce an intermediate problem which we call {\sc PCP Vectors}.
The purpose of introducing this problem is to abstract out the prover-verifier formulation
before proving hardness of approximation in \P, 
very much like $\NP$-hardness of approximation reductions start from gap-3-SAT or {\sc Label Cover}.

\begin{defn}
[{\sc PCP-Vectors}] The input to this problem consists of two sets
of vectors $A\subset\Sigma^{L\times K}$ and $B\subset\Sigma^{L}$,
The goal is to find vectors $a\in A$ and $b\in B$ that maximize
\begin{gather}\label{eq:score1}
s\left(a,b\right)\triangleq \Pr_{\ell \in L}\left[\bigvee_{k \in K}  \left(a_{\ell,k} = b_{\ell}\right)\right].
\end{gather}
\end{defn}

\begin{thm}
\label{thm:main}Let $\varepsilon>0$ be any constant, and
let $\left(A,B\right)$ be an instance of {\sc PCP-Vectors} with
$N$ vectors and parameters $|L|,|K|,|\Sigma|=N^{o\left(1\right)}$.
Then, assuming SETH, $O\left(N^{2-\varepsilon}\right)$-time
algorithms cannot distinguish between:
\begin{itemize}
\item (Completeness) there exist $a^{*},b^{*}$ such that $s\left(a^{*},b^{*}\right)=1$;
and
\item (Soundness) for every $a\in A,b\in B$, we have $s\left(a,b\right) \leq 1 /2^{(\log N)^{1-o(1)}}$.
\end{itemize}
\end{thm}

\fi

\paragraph{Bichromatic Max Inner Product}
Our first application is a strong resolution of Open Question~\ref{oq1}, under SETH.
Not only is an $O(1)$-factor approximation impossible in $O(N^{1+\eps})$ time, but we must pay a near-polynomial $2^{(\log{N})^{1-o(1)}}$ approximation factor if we do not spend nearly-quadratic $N^{2-o(1)}$ time! (See Theorem~\ref{thm:subset} below.)

As we mentioned earlier, when viewing the $\{0,1\}^d$ vectors as subsets of $[d]$, {\sc Max-IP} corresponds to maximizing the size of the intersection. 
In fact our hardness of approximation result holds even in a seemingly easier special case of {\sc Max-IP} which has received extensive attention: the Subset Query problem \cite{Rama+00,MG03,AAK10,GG10}, which is known to be equivalent to the classical Partial Match problem. The first non-trivial algorithms for this problem appeared in Ronald Rivest's PhD thesis~\cite{Rivest74,Rivest76}.
Since our goal is to prove lower bounds, we consider its offline or batch version (and the lower bound will transfer to the data structure version):
\begin{quote}
Given a collection of (text) sets $T_1,\ldots,T_N \subseteq [d]$ and a collection of (pattern) sets $P_1,\ldots,P_N \subseteq [d]$, is there a set $P_i$ that is contained in a set $T_j$?
\end{quote}
In the $c$-approximate case, we want to distinguish between the case of exact containment, and the case where no $T_j$ can cover more than a $c$-fraction of any $P_i$. We show that even in this very simple problem, we must pay a $2^{(\log{N})^{1-o(1)}}$ approximation factor if it is to be solved in truly-subquadratic time. Hardness of approximation for {\sc Max-IP} follows as a simple corollary of the following stronger statement:

\begin{thm}
\label{thm:subset}
Assuming SETH, for any $\eps>0$,  given two collections $A,B$, each of $N$ subsets of a universe $[m]$, where $m=N^{o(1)}$ and all subsets $b \in B$ have size $k$, no $O(N^{2-\eps})$ time algorithm can distinguish between the cases:

\ifFULL

\begin{description}
\item [{Completeness}] there exist $a \in A, b \in B$ such that $b \subseteq a$;
and
\item [{Soundness}] for every $a \in A, b \in B$ we have $| a \cap b | \leq  k/2^{(\log{N})^{1-o(1)}}$.
\end{description}

\else
\begin{itemize}
\item (Completeness) there exist $a \in A, b \in B$ such that $b \subseteq a$;
and
\item (Soundness) for every $a \in A, b \in B$ we have $| a \cap b | \leq  k/2^{(\log{N})^{1-o(1)}}$.
\end{itemize}

\fi

\end{thm}

Improving our lower bound even to some $N^{\eps}$ factor (for a universal constant $\eps>0$) would refute SETH via the known {\sc Max-IP} algorithms (see e.g.~\cite{pods16}). 
Using an idea of~\cite{WW10-subcubic}, it is not hard to show that Theorem~\ref{thm:subset} also applies to the harder (but more useful) search version widely known as MIPS.

\begin{cor}
\label{cor:MIPS}
Assuming SETH, for all $\eps>0$, no algorithm can preprocess a set of $N$ vectors $p_1,\ldots,p_N \in D \subseteq \{0,1\}^m$ in polynomial time, and subsequently given a query vector $q \in \{0,1\}^m$ can distinguish in $O(N^{1-\eps})$ time between the cases:

\ifFULL

\begin{description}
\item [{Completeness}] there exist $p_i \in D$ such that $\langle p_i, q \rangle \geq s$;
and
\item [{Soundness}] for every $p_i \in D$, $\langle p_i, q \rangle \leq s/2^{(\log{N})^{1-o(1)}}$,
\end{description}
even when $m=N^{o(1)}$ and the similarity threshold $s \in [m]$ is fixed for all queries.
\end{cor}

\else
\begin{itemize}
\item (Completeness) there is a $p_i \in D$ such that $\langle p_i, q \rangle \geq s$;
and
\item (Soundness) for all $p_i \in D$, $\langle p_i, q \rangle \leq s/2^{(\log{N})^{1-o(1)}}$,
\end{itemize}
even when $m=N^{o(1)}$ and the similarity threshold $s \in [m]$ is fixed for all queries $q$.
\end{cor}

\fi

Except for the $(1+o(1))$-factor lower bound \cite{pods16} which transfers to MIPS as well, the only lower bounds known were either for specific techniques \cite{MNP07,ACP08,OWZ14,AILRS15}, or were in the cell-probe model but only ruled out extremely efficient queries \cite{AIP06,PTW08,PTW10,KP12,AV15,ALRW16}. 

An important version of {\sc Max-IP} is when the vectors are in $\{ -1, 1 \}^d$ rather than $\{0,1\}^d$. This version is closely related to other famous problems such as the light bulb problem and the problem of learning parity with noise (see the reductions in \cite{Valiant15}).
Negative coordinates often imply trivial results for \emph{multiplicative} hardness of approximation: it is possible to shift a tiny gap of $k$ vs. $k+1$ to a large multiplicative gap of $0$ vs $1$ by adding $k$ coordinates with $-1$ contribution.
In the natural version where we search for a pair with maximum inner product \emph{in absolute value}, this trick does not quite work.
Still, Ahle et al.~\cite{pods16} exploit such cancellations to get a strong hardness of approximation result using an interesting application of Chebychev embeddings.
The authors had expected that a different approach must be taken to prove constant factor hardness for the $\{0,1\}$ case.
Interestingly, since it is easy%
\ifFULL
\else
\footnote{E.g. map each $0$ to a random string in $\{\pm 1\}^d$, and map each $1$ to the string $1^d$.}, 
\fi
~to reduce $\{0,1\}$ to $\{-1,1\}$,
our reduction also improves their lower bound for the $\{-1,1\}$ case from $2^{\tilde \Omega(\sqrt{\log{N}})}$ to the almost-tight $2^{(\log{N})^{1-o(1)}}$
\ifFULL
(see Corollary~\ref{cor:signed-maxip}). 
This also implies an $N^{1-o(1)}$-time lower bound for queries in the indexing version of the problem.

\paragraph{Longest Common Subsequence Closest Pair}
Efficient approximation algorithms have the potential for major impact in \emph{sequence alignment problems}, the standard similarity measure between genomes and biological data. One of the most cited scientific papers of all time studies BLAST, a \emph{heuristic} algorithm for sequence alignment that often returns grossly sub-optimal solutions\footnote{Note that many of its sixty-thousand citations are by other algorithms achieving better results (on certain datasets).} but always runs in near-linear time, in contrast to the best-known worst-case quadratic-time algorithms.
For theoreticians, to get the most insight into these similarity measures, it is common to think of them as Longest Common Subsequence (LCS) or Edit Distance.
The {\sc Bichromatic LCS Closest Pair} problem is:
\begin{quote} Given a (data) set $N$ strings and a (query) set of $N$ strings, all of which have length $m \ll N$, find a pair, one from each set, that have the maximum length common subsequence (noncontiguous).   
\end{quote}
The search version and the Edit Distance version are defined analogously. Good algorithms for these problems would be highly relevant for bioinformatics.

A series of breakthroughs \cite{landau1998incremental,Indyk04-product,bar2004approximating,batu2006oblivious,ostrovsky2007low,andoni2010polylogarithmic,andoni2012approximating} led to ``good'' approximation algorithms for Edit Distance.
While most of these papers focus on the more basic problem of approximating the Edit Distance of two given strings (see Open Question~\ref{oq:2strings}), they use techniques such as metric embeddings that can solve the Edit Distance Closest Pair problem in near-linear time with a $2^{O(\sqrt{\log{m} \log\log{m}})}$ approximation \cite{ostrovsky2007low}.
Meanwhile, both the basic LCS problem on two strings and the LCS Closest Pair resisted all these attacks, and to our knowledge, no non-trivial algorithms are known.
On the complexity side, only a $(1+o(1))$-approximation factor lower bound is known for {\sc Bichromatic LCS Closest Pair}~\cite{ABV15a,BK15,AB17}, and getting a $1.001$ approximation in near-linear time is not known to have any consequences.
For algorithms that only use metric embeddings there are nearly logarithmic lower bounds for Edit-Distance Closest Pair, but even under such restrictions the gaps are large \cite{ADG+03,SU04,KN05,AK07,AJP10}.

Perhaps our most surprising result is a \emph{separation} between these two classical similarity measures. Although there is no formal equivalence between the two, they have appeared to have the same complexity no matter what the model and setting are.
We prove that LCS Closest Pair is \emph{much harder} to approximate than Edit Distance.

\begin{thm}
\label{thm:lcs}
Assuming SETH, there is no $\left({2^{(\log{N})^{1-o(1)}}}\right)$-approximation algorithm for {\sc Bichromatic LCS Closest Pair} on two sets of $N$ permutations of length $m=N^{o(1)}$ in time $O\left(N^{2-\varepsilon}\right)$, for all $\varepsilon>0$. 
\end{thm}

Notice that our theorem holds even under the restriction that the sequences are \emph{permutations}.
This is significant: in the ``global'' version of the problem where we want to compute the LCS between two long strings of length $n$, one can get the exact solution in near-linear time if the strings are permutations (the problem becomes the famous Longest Increasing Subsequence), while on arbitrary strings there is an $N^{2-o(1)}$ time lower bound from SETH.
The special case of permutations has received considerable attention due to connections to preference lists, applications to biology, and also as a test-case for various techniques.
In 2001, Cormode, Muthukrishnan, and Sahinalp \cite{CMS01} gave an $O(\log{m})$-Approximate Nearest Neighbor data structure for Edit Distance on permutations with $N^{o(1)}$ time queries (improved to $O(\log\log{m})$ in \cite{AIK09}), and raised the question of getting similar results for LCS Approximate Nearest Neighbor.
Our result gives a strong negative answer under SETH, showing that {\sc Bichromatic LCS Closest Pair} suffers from near-polynomial approximation factors when the query time is truly sublinear.

\paragraph{Regular Expression Matching}
Given two sets of strings of length $m$, a simple hashing-based approach lets us decide in near-linear time if there is a pair of Hamming distance $0$ (equal strings), or whether all pairs have distance at least $1$.
A harder version of this problem, which appears in many important applications, is when one of the sets of strings is described by a regular expression:
\begin{quote}
Given a regular expression $R$ of size $N$ and a set $S$ of $N$ strings of length $m$, can we distinguish between the case that some string in $S$ is matched by $R$, and the case that every string in $S$ is far in Hamming distance%
\footnote{In our hard instances, all the strings in $\Lreg(R)$ will be of length $m$, so Hamming distance is well defined.}
 from every string in $\Lreg(R)$ (the language defined by $R$)?
\end{quote}
This is a basic approximate version of the classical regular expression matching problem that has been attacked from various angles throughout five decades, e.g. \cite{Tho68,myers1989approximate,Mye92,wu1995subquadratic,knight1995approximate,myers1998reporting,navarro2004approximate,belazzougui2013approximate,BT09,BI16,BGL16}.
Surprisingly, we show that this problem is essentially as hard as it gets: even if there is an exact match, it is hard to find any pair with Hamming distance $(1-\varepsilon)\cdot m$, for any $\varepsilon>0$.
For the case of binary alphabets, we show that even if an exact match exists (a pair of distance $0$), it is hard to find a pair of distance $(\frac{1}{2}-\varepsilon) \cdot m$, for any $\varepsilon>0$.
Our lower bounds also rule out interesting algorithms for the harder setting of Nearest-Neighbor queries:
Preprocess a regular expression so that given a string, we can find a string in the language of the expression that is approximately-the-closest one to our query string.
The formal statement and definitions of regular expressions are given in 
\ifFULL
Section~\ref{sec:regexp}.
\else
the full version.
\fi

\begin{thm}[informal]
Assuming SETH, no $O(N^{2-\eps})$-time algorithm can, given a regular expression $R$ of size $N$ and a set $S$ of $N$ strings of length $m=N^{o(1)}$, distinguish between the two cases:

\ifFULL

\begin{description}
\item [{Completeness}] some string in $S$ is in $\Lreg(R)$;
and
\item [{Soundness}] all strings in $S$ have Hamming distance $(1-o(1))\cdot m$ (or, $(\frac{1}{2}-o(1))\cdot m$ if the alphabet is binary) from all strings in $\Lreg(R)$.
\end{description}

\else

\begin{itemize}
\item (Completeness) some string in $S$ is in $\Lreg(R)$
\item (Soundness) all strings in $S$ have Hamming distance $(1-o(1))\cdot m$ (or, $(\frac{1}{2}-o(1))\cdot m$ if the alphabet is binary) from all strings in $\Lreg(R)$.
\end{itemize}

\fi

\end{thm}

\paragraph{Diameter in Product Metrics}

The {\em diameter} (or {\em furthest pair}) problem has been well-studied in a variety of metrics (e.g. graph metrics \cite{ACIM99,RV13,CLRSTW14}).
There is a trivial $2$-approximation in near-linear time (return the largest distance from an arbitrary point), and for arbitrary metrics (to which we get query access) there is a lower bound stating that a quadratic number of queries is required to get a $(2-\delta)$-approximation \cite{Indyk99-diameter_lb}.
For $\ell_2$-metric, there is a sequence of improved subquadratic-time approximation algorithms~\cite{EK89-diameter_sqrt3, FP02-diameter, BOR04-diameter, Indyk00-diameter, GIV01-diameter, Indyk03-diameter}. 
The natural generalization to the $\ell_p$-metric for arbitrary $p$ is, to the best of our knowledge, wide open.

While we come short of resolving the complexity of approximating the diameter for 
$\ell_p$-metrics, 
we prove a tight inapproximability result for the slightly more general problem for the product of $\ell_p$ metrics.
Given a collection of metric spaces $M_i = \langle X_i,\Delta_i \rangle$, their $f$-{\em product metric} is defined as 
\ifFULL

$$\Delta\Big((x_1, \ldots, x_k), (y_1, \ldots, y_k) \Big) \triangleq f\Big(\Delta_1(x_1, y_1), \ldots, \Delta_k(x_k, y_k) \Big).$$

\else
$$\Delta\Big((x_1, \ldots, x_k), (y_1, \ldots, y_k) \Big)$$
$$ \triangleq f\Big(\Delta_1(x_1, y_1), \ldots, \Delta_k(x_k, y_k) \Big).$$
\fi
In particular, we are concerned with the $\ell_2$-product of $\ell_{\infty}$-spaces, whose metric is defined as:
\begin{gather} \label{eq:product}
\Delta_{2,\infty}(x,y) \triangleq 
	\sqrt{\sum_{i =1}^{d_2} \left( 
		\max_{j = 2}^{d_{\infty}} \left\{
				\Big| x_{i,j} - y_{i,j} \Big|
				\right\} \right)^2}.
\end{gather}
(This is a special case of the more general $\Delta_{2,\infty,1}(\cdot,\cdot)$ product metric, studied by~\cite{AIK09}.)

Product metrics (or {\em cascaded norms}) are useful for aggregating different types of data~\cite{Indyk98-product, Indyk02-product, CM05-cascaded, JW09-cascaded}.
They also received significant attention from the algorithms community because they allow rich embeddings, yet are amenable to algorithmic techniques (e.g. \cite{Indyk02-product, Indyk03-diameter, Indyk04-product, AIK09,AJP10, andoni2012approximating}).

\begin{thm}[Diameter]
Assuming SETH, there are no $(2-\delta)$-approximation algorithms for {\sc Product-Metric Diameter} in time $O\left(N^{2-\varepsilon}\right)$, for any constants $\varepsilon,\delta>0$.
\end{thm}

\subsection{Related work}

For all the problems we consider, SETH lower bounds for the exact version are known.
See \cite{Wil05,AW15} for the {\sc Max-IP} and {\sc Subset Queries} problems, \cite{AVW14,ABV15a,BK15,AHVW16} for {\sc Bichromatic LCS Closest Pair}, \cite{BI16,BGL16} for {\sc Regular Expression Matching}, and \cite{Wil05} for {\sc Metric Diameter}.
Also,~\cite{Wil18} recently proved SETH lower bounds for the related problems of exact $\ell_2$ {\sc Diamater} and {Bichromatic Closest Pair} over short vectors ($d= \poly \log \log (N)$).

Prior to our work, some hardness of approximation results were known using more problem-specific techniques.
For example, distinguishing whether the diameter of a graph on $O(n)$ edges is $2$ or at least $3$ in truly-subquadratic time refutes SETH \cite{RV13}, which implies hardness for $(3/2-\eps)$ approximations. (This is somewhat analogous to the $\NP$-hardness of distinguishing $3$-colorable graphs from graphs requiring at least $4$ colors, immediately giving hardness of approximation for the chromatic number.) 
In most cases, however, this fortunate situation does not occur.
The only prior SETH-based hardness of approximation results proved with more approximation-oriented techniques are by Ahle et al. \cite{pods16} for {\sc Max-IP} via clever embeddings of the vectors. 
As discussed above, for the case of $\{0,1\}$-valued vectors, their inapproximability factor is still only $1+o(1)$.

\cite{AB17} show that, under certain complexity assumptions, {\em deterministic} algorithms cannot approximate the Longest Common Subsequence (LCS) of two strings to within $1+o(1)$ in truly-subquadratic time.
They tackle a completely orthogonal obstacle to proving SETH-based hardness of approximation: for problems like LCS with two long strings, the quality of approximation depends on the {\em fraction of assignments} that satisfy a SAT instance. There is a trivial algorithm for approximating this fraction: sample assignments uniformly at random.  See further discussion on Open Question~\ref{oq:2strings}. 

Recent works by Williams~\cite{Williams16_MA-SETH} (refuting the \MA-variant of SETH) and Ball et al.~\cite{BRSV17-average_case} also utilize low-degree polynomials in the context of SETH and related conjectures. Their polynomials are quite different from ours: they sum over many possible assignments, and are hard to {\em evaluate} (in contrast, the polynomials used in 
\ifFULL
the proof of our Theorem~\ref{thm:CC} 
\else
our proof
\fi
correspond to a single assignment, and they are trivial to evaluate).

The main technical barrier to hardness of approximation in \P~is the blowup incurred by standard PCP constructions; in particular, we overcome it with distributed constructions. 
There is also a known construction of PCP with linear blowup for large (but sublinear) query complexity~\cite{BKKMS16-linear_PCP} with non-uniform verifiers; note however that merely obtaining linear blowup is not small enough for our purposes. Different models of ``non-traditional'' PCPs, such as interactive PCPs~\cite{KR08-interactive_pcp} and interactive oracle proofs (IOP)~\cite{BCS16-IOP, RRR16-IOP} have been considered and found ``positive'' applications in cryptography (e.g.~\cite{GKR15-delegation_muggles, GIMS10-zkPCP, BCS16-IOP}). 
In particular,~\cite{BCGRS16-IOP} obtain a linear-size IOP. It is an open question whether these interactive variants can imply interesting hardness of approximation results~\cite{BCGRS16-IOP}.
(And it would be very interesting if our distributed PCPs have any cryptographic applications!)

After the first version of this paper became public, it was brought to our attention
that the term "distributed PCP" has been used before in a different context by Drucker \cite{Drucker10_thesis}.
In the simplest variant of Drucker's model, Alice and Bob want to compute $f(\alpha, \beta)$ with minimal communication.
They receive a PCP that allegedly proves that $f(\alpha, \beta) = 1$; 
Alice and Bob each query the PCP at two random locations and independently decide whether to accept or reject the PCP. 
As with the interactive variants of PCP, we don't know of any implications of Drucker's work for hardness of approximation, but we think that this is a fascinating research direction.

\subsection{Discussion}
In addition to resolving the fine-grained approximation complexity of several fundamental problems, 
our work opens a hope to understanding more basic questions in this area.
We list a few that seem to represent some of the most fundamental challenges, as well as exciting applications.

\paragraph{{\sc Bichromatic LCS Closest Pair Problem} over $\{0,1\}$}
The {\sc Bichromatic LCS Closest Pair Problem} is most interesting in two regimes: 
permutations (which, by definition, require a large alphabet);
and small alphabet, most notably $\{0,1\}$.
For the regime of permutations, we obtain nearly-polynomial hardness of approximation. 
For small alphabet $\Sigma$, per contra, there is a trivial $1/|\Sigma|$-approximation algorithm in near-linear time: pick a random $\sigma \in \Sigma$, and restrict all strings to their $\sigma$-subset.
Are there better approximation algorithms?

Our current hardness techniques are limited because this problem does not admit an approximation preserving OR-gadget for a large OR. In particular the $1/|\Sigma|$-approximation algorithm outlined above implies that we cannot combine much more than $|\Sigma|$ substrings in a clever way and expect the LCS to correspond to just one substring. 
\begin{oq}\label{oq:binary}
Is there a $1.1$-approximation for the {\sc Bichromatic LCS Closest Pair Problem} on binary inputs running in $O(n^{2-\eps})$ time, for some $\eps>0$?
\end{oq}

\paragraph{{\sc LCS Problem} (with two strings)}
Gadgets constructed in a fashion similar to our proof of Theorem~\ref{thm:lcs} can be combined together (along with some additional gadgets) into two long strings $A,B$ of length $m$, in a way that yields a reduction from SETH to computing the longest common subsequence (LCS) of $(A,B)$, ruling out {\em exact} algorithms in $O(m^{2-\varepsilon})$~\cite{ABV15a,BK15}.
However, in the instances output by this reduction, approximating the value of the LCS reduces to approximating the {\em fraction} of assignments that satisfy the original formula; it is easy to obtain a good additive approximation by sampling random assignments.
The recent work of \cite{AB17} mentioned above, uses complexity assumptions on deterministic algorithms to tackle this issue, but their ideas do not seem to generalize to randomized algorithms.
\begin{oq}\label{oq:2strings}
Is there a $1.1$-approximation for {\sc LCS} running in $O(n^{2-\eps})$ time, for some $\eps>0$? (Open for all alphabet sizes.)
\end{oq}

\paragraph{Dynamic Maximum Matching}
A holy grail in dynamic graph algorithms is to maintain a $(1+\eps)$-approximation for the \emph{Maximum Matching} in a dynamically changing graph, while only spending amortized $n^{o(1)}$ time on each update.  
Despite a lot of attention in the past few years~
\cite{gupta2013fully,neiman2016simple,bhattacharya2014deterministic,
baswana2015fully,bernstein2015fully,bhattacharya2016design,bernstein2016faster,
peleg2016dynamic,bhattacharya2016new,solomon2016fully},
current algorithms are far from achieving this goal: one can obtain a $(1+\eps)$-approximation by spending $\Omega(\sqrt{m})$ time per update, or one can get an $2$-approximation with $\tilde{O}(1)$ time updates.

For exact algorithms, we know that $n^{o(1)}$ update times are impossible under popular conjectures \cite{Pat10,AV14,KPP16,HKNS15,Dahlgaard16}, such as $3$-SUM\footnote{The $3$-SUM Conjecture, from the pioneering work of \cite{GO95}, states that we cannot find three numbers that sum to zero in a list of $n$ integers in $O(n^{2-\eps})$ time, for some $\eps>0$.}, Triangle Detection\footnote{The conjecture that no algorithm can find a triangle in a graph on $m$ edges in $O(m^{4/3-\eps})$ time, for some $\eps>0$, or even just that $O(m^{1+o(1)})$ algorithms are impossible \cite{AV14}.} and the related Online Matrix Vector Multiplication\footnote{The conjecture that given a Boolean $n \times n$ matrix $M$ and a sequence of $n$ vectors $v_1,\ldots,v_n \in \{0,1\}^n$ we cannot compute the $n$ products $M\cdot x_i$ in an online fashion (output $Mx_i$ before seeing $x_{i+1}$) in a total of $O(n^{3-\eps})$ time \cite{HKNS15}. See \cite{LW17} for a recent upper bound.}.
From the viewpoint of PCP's, this question is particularly intriguing since it seems to require hardness amplification for one of these other conjectures.
Unlike all the previously mentioned problems, even the exact case of dynamic matching is not known to be SETH-hard.

\begin{oq}
Can one maintain an $(1+\eps)$-approximate maximum matching dynamically, with $n^{o(1)}$ amortized update time?
\end{oq}

\subsubsection*{New frameworks for hardness of approximation}
More fundamental than resolving any particular problem, 
our main contribution is a conceptually new framework for proving hardness of approximation for problems in \P~via {\em distributed PCPs}. 
In particular, we were able to resolve several open problems while relying on simple algebrization techniques from early days of PCPs (e.g.~\cite{LFKN92} and reference therein). 
It is plausible that our results can be improved by importing into our framework 
more advanced techniques from decades of work on PCPs --- starting with verifier composition~\cite{AS98-PCP}, parallel repetition~\cite{Raz98-parallel_repetition}, Fourier analysis~\cite{Hastad01}, etc.

\paragraph{Hardness from other sublinear communication protocols for Set Disjointness}
A key to our results is an \MA~protocol for Set Disjointness with sublinear communication, which trades off between the size of Merlin's message and the size of Alice and Bob's messages. There are other non-standard communication models where Set Disjointness enjoys a sublinear communication protocol, for example quantum communication%
\footnote{Note that due to Grover's algorithm, SETH is false for quantum computational complexity; but it is also false for \MA~\cite{Williams16_MA-SETH}, which doesn't prevent us from using the \MA~communication protocol in an interesting way.}~\cite{BCW98-quantum_communication}.
\begin{oq}
Can other communication models inspire new reductions (or algorithms) for standard computational complexity?
\end{oq}

\paragraph{Hardness of approximation from new models of PCPs}
This is the most open-ended question. 
Formulating a clean conjecture about distributed PCPs was extremely useful for understanding the limitations and possibilities of our framework --- even though our original conjecture turned out to be false.
\begin{oq}
Formulate a simple and plausible PCP-like conjecture that resolves any of the open questions mentioned in this section. 
\end{oq}

\subsubsection*{Bichromatic vs Monochromatic Closest Pair}
Essentially all known SETH-based hardness results for variants of closest pair, including the ones in this paper (Theorems~\ref{thm:subset} and~\ref{thm:lcs}), hold for the so called ``bichromatic'' variant: the input to the algorithm is two sets $A,B$, and the goal is to find the closest pair $(a,b) \in A\times B$.
Indeed, this is the most interesting variant because it implies hardness for the corresponding data structure variants (as in Corollary~\ref{cor:MIPS}). 
Surprisingly, this understanding of bichromatic closest pair problems does not seem to translate to the corresponding ``monochromatic'' variants, where the input to the algorithm is a single set $U$, and the goal is to find the closest pair $u,v \in U$ ($u \neq v$).
In an earlier version of this paper we had mistakenly claimed that our techniques here can also shed light on this problem, but this turned out to be incorrect (see also Erratum in Section~\ref{sec:erratum}).
Understanding the complexity of exact and approximate monochromatic closest pair problems for many metrics remains an interesting open question; see e.g.~\cite{DSL16-bichromatic, Wil18} for further discussion.

\ifFULL
\section{Preliminaries}\label{sec:prelim}

\paragraph{The Strong Exponential Time Hypothesis} 
SETH was suggested by Impagliazzo and Paturi~\cite{CIP06,IP01-SETH} as a possible explanation for the lack of improved algorithms for {\sc $k$-SAT}: Given a $k$-CNF formula (each clause has $k$ literals) on $n$ variables, decide if it is satisfiable.
Known algorithms have an upper bound of the form $O(2^{(1-c/k)n})$, where $c$ is some constant, which makes them go to $2^n$ as $k$ grows.
The conjecture is that this is inevitable.
\begin{conj}[SETH]
\label{conj:SETH}%
For any $\eps > 0$ there exists $k \ge 3$ such that {\sc $k$-SAT} on $n$ variables cannot be solved in time $O(2^{(1-\eps)n})$.
\end{conj}

\paragraph{Sparsification Lemma} 
An important tool for working with SETH is the \emph{Sparsification Lemma} of Impagliazzo, Paturi, and Zane~\cite{IPZ01-ETH}, which implies the following.
\begin{lem}[Follows from Sparsification Lemma~\cite{IPZ01-ETH}]
\label{lem:sparsification}%
If there is an $\eps>0$ such that for all $k \ge 3$ we can solve  {\sc $k$-SAT} on $n$ variables and $c_{k,\eps}\cdot n$ clauses in $O(2^{(1-\eps)n})$ time, then SETH is false.
 \end{lem}

\paragraph{Notation}
We use $[n]$ to denote the set $\{1,\dots,n\}$, and $x_{-n}$ to denote the vector $(x_1 ,\dots ,x_{n-1})$.

\section{\MA~Communication Complexity}\label{sec:CC}

In this section we prove a \MA~(Merlin-Arthur) style communication protocol for Set Disjointness with sublinear communication and near-polynomial soundness.
Our protocol is similar to the protocol from~\cite{AW09-algebrization}, but optimizes different parameters (in particular, we obtain very low soundness).

\begin{thm}\label{thm:CC}
For $T \leq n$, there exists a computationally efficient \MA-protocol for Set Disjointness in which:
\begin{enumerate}
\item Merlin sends Alice $O(n \log{n}/T)$ bits; 
\item Alice and Bob jointly toss  $\log_2 n +O(1)$ coins; 
\item Bob sends Alice $O(T \log{n})$ bits.
\item Alice returns \textup{Accept} or \textup{Reject}.
\end{enumerate}
If the sets are disjoint, there is a message from Merlin such that Alice always accepts; 
otherwise (if the sets have a non-empty intersection or Merlin sends the wrong message), Alice rejects with probability $\geq \frac{1}{2}$. 
\end{thm}
\begin{proof}
\subsubsection*{Arithmetization}

Assume without loss of generality that T divides $n$.
Let $q$ be a prime number with $4n \leq q \leq 8n$,
and let $\mathbb{F}_q$ denote the prime field of size $q$.
We identify  $[\frac{n}{T}]\times[T]$ with the universe $[n]$ over which Alice and Bob want to compute Set Disjointness. 

Alice's input can now be represented as $T$ functions
$\psi_{\alpha,t}:[\frac{n}{T}]\rightarrow\left\{ 0,1\right\} $
as follows: $\psi_{\alpha,t}(i)\triangleq 1$ if and only if the element corresponding to $(i,t)$ is in Alice's set.  Define
$\psi_{\beta,t}$ analogously. Notice that their sets are disjoint  
if and only if $\psi_{\alpha,t}(i)\psi_{\beta,t}(i)=0$ for all
$ i,t \in[\frac{n}{T}]\times[T]$.

Extend each $\psi_{\alpha,t},\psi_{\beta,t}$ to 
polynomials $\Psi_{\alpha,t},\Psi_{\beta,t}:\mathbb{F}_q\rightarrow \mathbb{F}_q$
of degree at most $\frac{n}{T}-1$ in every variable.
Notice that the polynomials $\Psi_t(i)\triangleq\Psi_{\alpha,t}(i)\cdot\Psi_{\beta,t}(i)$
have degree at most $2(\frac{n}{T}-1)$. 
The same degree bound also holds for the sum of those polynomials, 
$\Psi \triangleq \sum_{t=1}^T \Psi_t$. 
Notice that the sets are disjoint if and only if $\Psi_t(i)=0$ for all $i,t \in [\frac{n}{T}]\times[T]$.
Similarly%
\footnote{Here we use the fact that $\mathbb{F}_q$ has a large characteristic,
and $\Psi_t(i)\in\left\{ 0,1\right\} $ for all
$i,t\in [\frac{n}{T}]\times[T]$; so the summation (in $\mathbb{F}_q$)
of $T$ zeros and ones is equal to zero if and only if there are no ones.%
}, 
the sets are disjoint
if and only if $\Psi(i)=0$ for all $ i \in [\frac{n}{T}]$.

\subsubsection*{The protocol}
We begin with a succinct formal description of the protocol, and provide more details below.

\begin{enumerate}
\item Merlin sends Alice $\Phi$ which is allegedly equal to the marginal sums polynomial $\Psi$.
\item Alice and Bob jointly draw $i \in \mathbb{F}_q$ uniformly at random.
\item Bob sends Alice $\Psi_{\beta,t}(i)$ for every $t \in [T]$.
\item Alice accepts if and only if both of the following hold:
\begin{gather}
\forall i \in \left[\frac{n}{T} \right] \;\;\ \Phi(i)=0 \label{eq:Phi}\\
 \Phi(i) = \sum_{t = 1}^T\Psi_{\alpha,t}(i) \cdot \Psi_{\beta,t}(i).\label{eq:test2}
\end{gather}
\end{enumerate}

Recall that Merlin knows both $\Psi_{\alpha,t}$ and $\Psi_{\beta,t}$. In the first step of the protocol, Merlin sends Alice 
a polynomial $\Phi$ which is allegedly equal to $\Psi$. Notice that $\Psi$ is a (univariate) polynomial of degree at most $2(\frac{n}{T}-1)$; thus it can be uniquely specified by $2 \frac{n}{T}-1$ coefficients in $\mathbb{F}_q$. 
Since each coefficient only requires $\log_2{|\mathbb{F}_q|} = \log_2 n +O(1)$ bits, the bound on Merlin's message follows.

In the second step of the protocol Alice and Bob draw $i \in \mathbb{F}_q$ uniformly at random. 
In the third step of the protocol, Bob sends Alice the values of $\Psi_{\beta,\cdot}(i)$. 
In particular sending $T$ values in $\mathbb{F}_q$ requires $O(T \log{n})$ bits.

\subsubsection*{Analysis}

\begin{description}
\item[Completeness]
If the sets are disjoint, Merlin can send the true $\Psi$, and Alice always accepts.
\item[Soundness]
If the sets are not disjoint, Merlin must send a different low degree polynomial (since \eqref{eq:Phi} is false for $\Psi$).
By the Schwartz-Zippel Lemma, since $\Psi$ and $\Phi$ have degree less than $2 \frac{n}{T} \leq \frac{q}{2}$, if they are distinct they must differ on at least half of $\mathbb{F}_q$. Hence, \eqref{eq:test2} is false with probability $\geq \frac{1}{2}$. 
(The same holds if the sets are disjoint but Merlin sends the wrong $\Phi$.)
\end{description}
\end{proof}

In the following corollary we amplify the soundness of the communication protocol, to obtain stronger computational hardness results.

\begin{cor}\label{cor:CC}
There exists a computationally efficient \MA-protocol for Set Disjointness s.t.:
\begin{enumerate}
\item Merlin sends Alice $o(n)$ bits; \item Alice and Bob jointly toss  $o(n)$ coins; 
\item Bob sends Alice $o(n)$ bits.
\item Alice returns \textup{Accept} or \textup{Reject}.
\end{enumerate}
If the sets are disjoint, there is a unique message from Merlin such that Alice always accepts; 
otherwise (if the sets have a non-empty intersection or Merlin sends the wrong message), Alice rejects with probability $\geq 1-1/2^{n^{1-o(1)}}$. 
\end{cor}
\begin{proof}
Let $T$ be a small super-logarithmic function of $n$, e.g. $T = \log^2 n$. 
Repeat the protocol from Theorem~\ref{thm:CC} $R = \frac{n}{T^2}$ times to amplify the soundness.
Notice that since Merlin sends her message before the random coins are tossed, it suffices to only repeat steps 2-4.
Thus Merlin still sends $O(\frac{n}{T} \cdot \log{n}) = o(n)$ bits, Alice and Bob toss a total of $O(R \cdot \log n) = o(n)$ coins, and Bob sends a total of $O(R \cdot T  \cdot\log{n}) = o(n)$ bits.
\end{proof}

\begin{remark}
An alternative way to obtain Corollary~\ref{cor:CC} is via a ``white-box'' modification of the protocol from Theorem~\ref{thm:CC}: all the polynomials remain the same, but we consider their evaluation over an extension field $\mathbb{F}_{q^R}$. Note that Merlin's message (polynomial) remains the same as the coefficients are still in $\mathbb{F}_q$, but Bob's message is now $R$-times longer (as in the proof of Corollary~\ref{cor:CC}).
\end{remark}


\section{A distributed and non-deterministic PCP theorem}\label{sec:PCP}

In this section we prove our distributed, non-deterministic PCP theorem.

\begin{thm}[Distributed, Non-deterministic PCP Theorem]\label{thm:PCP}
Let $\varphi$ be a Boolean CNF formula with $n$ variables and $m = O(n)$ clauses.
There is a non-interactive protocol where:
\begin{itemize}
\item Alice, given the CNF $\varphi$, partial assignment $\alpha \in \{0,1\}^{n/2}$, and advice $\mu \in \{0,1\}^{o(n)}$, outputs a string $a^{\alpha,\mu} \in \{0,1\}^{2^{o(n)}}$.
\item Bob, given $\varphi$ and partial assignment $\beta \in \{0,1\}^{n/2}$, outputs a string $b^{\beta} \in \{0,1\}^{2^{o(n)}}$.
\item The verifier, given input $\varphi$, tosses $o(n)$ coins, non-adaptively reads $o(n)$ bits from $b^{\beta}$, and adaptively reads one bit from $a^{\alpha,\mu}$; finally, the verifier returns \textup{Accept} or \textup{Reject}.
\end{itemize}
If the combined assignment $(\alpha, \beta)$ satisfies $\varphi$, there exists advice $\mu^{*}$ such that the verifier always accepts. 
Otherwise (in particular, if $\varphi$ is unsatisfiable), for every $\mu$, the verifier rejects with probability $\geq1-1/2^{n^{1-o(1)}}$.
\end{thm}

\begin{proof}
For any partial assignments $\alpha,\beta\in\left\{ 0,1\right\} ^{n/2}$,
we consider the induced sets $S_{\alpha}, T_{\beta} \subseteq [m]$, where $j \in S_{\alpha}$ iff {\em none} of the literals in the $j^{th}$ clause receive a positive assignment from $\alpha$ (i.e. each literal is either set to false, or not specified by the partial assignment). Define $T_{\beta}$ analogously.
Notice that the joint assignment $(\alpha,\beta)$ satisfies $\varphi$ iff the corresponding sets are disjoint.
The construction below implements the \MA~communication protocol for Set Disjointness from Corollary~\ref{cor:CC} with inputs $S_{\alpha}, T_{\beta}$.

\subsubsection*{Constructing the PCP}

Bob's PCP is simply a list of all messages that he could send on the \MA~communication protocol, depending on the random coin tosses.
Formally, let $L$ enumerate over the outcomes of Alice and Bob's coin tosses in the protocol ($|L| \leq 2^{o(n)}$).
For each $\ell \in L$, we let ${b^{\beta}}_\ell$ be the message Bob sends on input $T_{\beta}$ and randomness $\ell$.

Alice's PCP is longer: she writes, for each possible outcome of the coin tosses, a list of all messages from Bob that she is willing to accept.
Formally, let $K$ enumerate over all possible Bob's messages (in particular,  $|K| = 2^{o(n)}$).
For each $k \in K$, if Alice accepts given message $\mu$ from Merlin, randomness $\ell$, and message $k$ from Bob, we set ${a^{\alpha, \mu}}_{\ell,k} \triangleq 1$.
Otherwise (if Alice rejects), we set ${a^{\alpha, \mu}}_{\ell,k} \triangleq 0$.

The verifier chooses $\ell \in L$ at random, reads ${b^{\beta}}_\ell$  and then ${a^{\alpha, \mu}}_{\ell,{b^{\beta}}_\ell}$ (i.e. accesses ${a^{\alpha, \mu}}_{\ell,\cdot}$ at index ${b^{\beta}}_\ell$). The verifier accepts iff 
$${a^{\alpha, \mu}}_{\ell,{b^{\beta}}_\ell} = 1.$$

\subsubsection*{Analysis}
Observe that for each $\ell$, we have that ${a^{\alpha, \mu}}_{\ell,{b^{\beta}}_\ell} = 1$ iff and Alice accepts (given Alice and Bob's respective inputs, Merlin's message $\mu$, randomness $\ell$, and Bob's message ${b^{\beta}}_\ell$). 
Therefore, the probability that the PCP verifier accepts is exactly equal to the probability that Alice accepts in the \MA~communication protocol.
\end{proof}


\section{PCP-Vectors}\label{sec:main}

In this section we introduce an intermediate problem which we call {\sc PCP Vectors}.
The purpose of introducing this problem is to abstract out the prover-verifier formulation
before proving hardness of approximation in \P, 
very much like $\NP$-hardness of approximation reductions start from gap-3-SAT or {\sc Label Cover}.

\begin{defn}
[{\sc PCP-Vectors}] The input to this problem consists of two sets
of vectors $A\subset\Sigma^{L\times K}$ and $B\subset\Sigma^{L}$,
The goal is to find vectors $a\in A$ and $b\in B$ that maximize
\begin{gather}\label{eq:score1}
s\left(a,b\right)\triangleq \Pr_{\ell \in L}\left[\bigvee_{k \in K}  \left(a_{\ell,k} = b_{\ell}\right)\right].
\end{gather}
\end{defn}

\begin{thm}
\label{thm:main}Let $\varepsilon>0$ be any constant, and
let $\left(A,B\right)$ be an instance of {\sc PCP-Vectors} with
$N$ vectors and parameters $|L|,|K|,|\Sigma|=N^{o\left(1\right)}$.
Then, assuming SETH, $O\left(N^{2-\varepsilon}\right)$-time
algorithms cannot distinguish between:
\begin{description}
\item [{Completeness}] there exist $a^{*},b^{*}$ such that $s\left(a^{*},b^{*}\right)=1$;
and
\item [{Soundness}] for every $a\in A,b\in B$, we have $s\left(a,b\right) \leq 1 /2^{(\log N)^{1-o(1)}}$.
\end{description}
\end{thm}

\begin{proof}
Let $\varphi$  be a CNF formula with $n$ variables and $m$ clauses
(without loss of generality $m=\Theta\left(n\right)$ by the Sparsification Lemma~\cite{IPZ01-ETH}).
Let $B$ be the sets of vectors generated by Bob in the distributed, non-deterministic PCP from Theorem~\ref{thm:PCP}, where we think of each substring ${b^{\beta}}_{\ell} \in \{0,1\}^{o(n)}$ as a single symbol in $\Sigma$. 
Similarly, let $\widehat A$ be the set of vectors generated by Alice.
For each $\alpha, \mu$, we modify $\widehat{a}^{\alpha, \mu}$ as follows:
$$
{a^{\alpha, \mu}}_{\ell, k} \triangleq
\begin{cases}
k & \text{if ${\widehat{a}^{\alpha, \mu}}_{\ell, k} = 1$}\\
\perp & \text{if ${\widehat{a}^{\alpha, \mu}}_{\ell, k} = 0$}
\end{cases}
.$$
By definition, $s\left({a^{\alpha, \mu}}_{\ell,k}, {b^{\beta}}_{\ell} \right)$ is exactly equal to the probability that the verifier accepts.
\end{proof}

\section{Max Inner Product and Subset Queries}

In this section, we present our first application of our PCP-Theorem by giving an extremely simple reduction from our {\sc PCP-Vectors} problem to the {\sc Bichromatic Max Inner Product} problem (and its special case, the Subset Query problem) from Section~\ref{sec:intro}. 

\begin{proof}(of Theorem~\ref{thm:subset})
Given an instance $A',B'$ of {\sc PCP-Vectors} as in Theorem~\ref{thm:main} with parameters $|L|,|K|,|\Sigma|=N^{o(1)}$ we map each vector $a' \in A' \subseteq \Sigma^{L\times K}$ to a subset $a$ of $U = L \times \Sigma$ in the following natural way.
For all $\ell \in [L], k \in [K]$ and $\sigma \in \Sigma$ we set add $(\ell,\sigma)$ to $a$ iff there is a $k \in K$ such that $a'_{\ell,k}= \sigma$.
We map the vectors $b' \in B$ to a subset $b$ of $U$ by adding the element $(\ell,\sigma)$ to $b$ iff $b'_{\ell}=\sigma$.
Note that the universe has size $d=|L|\cdot |\Sigma|=N^{o(1)}$ and the sets $b$ have size $|L|$.

The completeness and soundness follow from the fact that for any two vectors $a'\in A', b' \in B'$:
$$
| a \cap b| = |L| \cdot s(a',b').
$$
\end{proof}

\subsection{MIPS}

Next, we show the corollary for the search version of the problem known as MIPS.
This reduction from the (offline) closest pair problem to the (online) nearest neighbor problem is generic and works for all the problems we discuss in this paper.
The proof is based on a well-known technique~\cite{WW10-subcubic}, but might be surprising when seen for the first time.

\begin{proof} (of Corollary~\ref{cor:MIPS})
Assume we have a data structure as in the statement, and we will show how to solve {\sc Bichromatic Max Inner Product} instances on $2N$ vectors as in Theorem~\ref{thm:subset}, refuting SETH.
Let $c$ be such that $O(n^c)$ is an upper bound on the (polynomial) preprocessing time of the data structure.
Set $x=1/2c$ and note that it is small but constant. 
We partition the set $B$ of vectors into $t=N^{1-x}$ buckets of size $N^x$ each $B_1,\dots,B_t$.
For each bucket $B_i$ for $i \in [t]$ we use our data structure to do the following:
\begin{enumerate}
\item Preprocess the $n=N^{x}$ vectors in $B_i$.
\item For each of our $N$ vectors $a \in A$, ask the query to see if $a$ is close to any vector in $B_i$.
\end{enumerate}
Observe that after these $t$ checks are performed, we are guaranteed to find the high-inner product pair, if it exists.
The total runtime is $N^{1-x}$ times the time for each stage:
$$
N^{1-x} \cdot ( n^c + N \cdot n^{1-\eps}) = N^{1+x(c-1)} + N^{2-\eps x}
$$
which is $O(N^{2-\eps'})$ for some constant $\eps'>0$, refuting SETH.
\end{proof}

\subsection{Maximum Inner Product over $\{-1,1\}$}
We shall now prove that our results extend to the variant of {\sc Bichromatic Max Inner Product}  where the vectors are in $\{-1,1\}^d$ and the goal is to maximize the {\em absolute value} of the inner product.

\begin{cor}\label{cor:signed-maxip}
Assuming SETH, the following holds for every constant $\eps>0$. Given two sets of vectors $A,B \in \{-1,1\}^d$, where $|A| = |B| = N$, and $d = N^{o(1)}$, any $O(N^{2-\eps})$ time algorithm for computing 
$\max_{\substack{a \in A \\ b \in B}} |a \cdot b|$ 
must have approximation factor at least $2^{(\log{N})^{1-o(1)}}$.
\end{cor}
\begin{proof}
Starting with a hard $(A', B')$ instance of {\sc Bichromatic Max Inner Product} over $\{0,1\}^{d/4}$, we construct $(A,B)$ as follows. Consider the following three vectors in $\{-1,1\}^4$:
\begin{align*}
\gamma_1 & \triangleq (1,1,1,1) \\
\alpha_0 & \triangleq (1, 1, -1, -1) \\
\beta_0 & \triangleq (1, -1, 1, -1).
\end{align*}
Notice that $\gamma_1 \cdot \gamma_1= 4$, but $\alpha_0 \cdot \gamma_1 = \gamma_1 \cdot \beta_0= \alpha_0 \cdot \beta_0 = 0$.

We replace each entry in each vector in $A',B'$ with one of $\gamma_1, \alpha_0, \beta_0$ as follows: all $1$'s are replaced by $\gamma_1$, all $0$'s in $A$-vectors are replaced by $\alpha_0$, and all $0$'s in $B$-vectors are replaced by $\beta_0$. 
For vectors $a'\in A', b'\in B'$, let $a,b \in \{-1,1\}^d$ denote the new vectors that result from the reduction. Observe that now $a \cdot b = 4 a'\cdot b'$. 
\end{proof}


\section{Permutation LCS}

In this section, we prove that closest pair under the Longest Common Subsequence (LCS) similarity measure cannot be solved in truly subquadratic time, even when allowed near-polynomial approximation factors, and even when the input strings are restricted to be permutations. As a direct corollary, we show that approximate nearest neighbor queries under LCS cannot be computed efficiently: one must spend time proportional to the number of strings in the database.

\begin{defn}
[{\sc Bichromatic LCS Closest Pair Problem}]
Given two sets of strings $X,Y \subseteq \Sigma^{m}$ over an alphabet $\Sigma$, the goal is to find a pair $x \in X, y \in Y$ that maximize $LCS(x,y)$.
\end{defn}

If the optimal solution to the LCS Closest Pair problem is $OPT$, a $c$-approximation algorithm is allowed to return any value $L$  such that $OPT/c \leq L \leq OPT$, corresponding to a pair of strings that are $c$-away from the closest pair.

The rest of this section is dedicated to the proof of Theorem~\ref{thm:lcs} from the Introduction.

\begin{proof}(of Theorem~\ref{thm:lcs})

We reduce from {\sc Max-IP} to {\sc Bichromatic LCS Closest Pair} over permutations.
Specifically, for each vector $u \in A \cup B \subseteq \{0,1\}^m$ in the {\sc Max-IP} instance and each index $i \in [m]$, we encode $u_i$ as a permutation $\pi^{u}(u_i)$ over sub-alphabet $\Sigma_i$
(where $\Sigma_i$ is unique to the index $i$ and independent of the vector $u$). Our encoding (shown below) has the following guarantee: 
if $u_i = v_i = 1$, then $LCS\big(\pi^{u}(u_i), \pi^{v}(v_i)\big) = |\Sigma_i|$, and otherwise 
$LCS\big(\pi^{u}(u_i), \pi^{v}(v_i)\big) = O(\sqrt{|\Sigma_i|} \log{N})$.  

Observe that this suffices to prove our theorem for sufficiently large $|\Sigma_i|$.
The strings in our {\sc Bichromatic LCS Closest Pair} instance will be simply the concatenation of all the sub-permutations:
$$
x(u) \triangleq \bigcirc_{i=1}^{d} \pi^{u}(u_i).
$$
Notice that those strings are indeed permutations. Furthermore, because the sub-permutations use disjoint alphabets, we have:
\begin{align}
LCS\big(x(u), x(v) \big) & = \sum_{i=1}^d LCS\big(\pi^{u}(u_i), \pi^{v}(v_i)\big) \nonumber \\
& = (u \cdot v) |\Sigma_i|  + (d - u \cdot v) O(\sqrt{|\Sigma_i|}\log{N}). \label{eq:LCS-gap}
\end{align}
When setting $|\Sigma_i| = 2^{(\log N)^{1-o(1)}} d \log N$, the first term in \eqref{eq:LCS-gap} dominates the second one, and the reduction from {\sc Max-IP} follows. 

\subsubsection{Embedding bits as permutations}
We need $N+1$ different permutations over each $\Sigma_i$: the $0$-bits of different vectors should be embedded into very far strings (in LCS ``distance''), but we use the same embedding for all the $1$-bits.
First, observe that we could simply use $N+1$ random permutations; with high probability every pair will have LCS at most $O(\sqrt{|\Sigma_i|} \log N)$. Below, we show how to match this bound with a simple deterministic construction.

We assume w.l.o.g. that $|\Sigma_i|$ is a square of a prime.
Let ${\cal F}$ be the prime field of cardinality $|{\cal F}| = \sqrt{|\Sigma_i|}$.
We consider permutations over ${\cal F}^2$, sorted by the lexicographic order.
For every polynomial $p : {\cal F} \rightarrow {\cal F}$ of degree $\leq \log{N}$, 
we construct a permutation $\pi_p : {\cal F}^2 \rightarrow {\cal F}^2$ as follows:
$$ \pi_p(i,j) \triangleq (j, i + p(j)).$$
To see that those are indeed permutations, notice that we can invert them:
$$ \pi^{-1}_p(j,k) \triangleq (k-p(j), j).$$
Note also that we now have more than enough ($|{\cal F}|^{\log N} \gg N$) different permutations.

Let us argue that these permutations are indeed far in (in LCS ``distance''). Fix a pair of polynomials $p, q$, and consider the LCS of the corresponding permutations.
Suppose that for some $i,i',j$, we matched 
$$ \pi_p(i,j) = (j, i+p(j)) = (j, i'+q(j)) = \pi_{q}(i',j). $$
Then, for this particular choice of $i,i'$ the above equality holds for only $d$ distinct $j \in T$.
Thus, we can add at most $d$ elements to the LCS until we increase either $i$ or $i'$. 
Since each of those increases only $|{\cal F}|$ times, we have that
$$LCS(\pi_p, \pi_q) =O(|{\cal F}| \cdot d) = O(\sqrt{|\Sigma_i|} \cdot \log{N}).$$
\end{proof}
 

\section{Regular Expressions}
\label{sec:regexp}

In this section we prove hardness for a much simpler metric, the \emph{Hamming Distance} between strings, but where one of the sets is defined by a regular expression.
Our result is essentially the strongest possible negative result explaining why there are no non-trivial algorithms for our problem.

A regular expression is defined over some alphabet $\Sigma$ and some additional operations such as $\{ |, \circ, *, +\}$.
Our results will hold even in the restricted setting where we only allow the operations $|$ and $\circ$ (OR and concatenation), and so we will restrict the attention to this setting.
The simplest definition is recursive:
If $e_1,e_2$ are two regular expressions, and $\sigma \in \Sigma$ is a letter, then: 
\begin{itemize}
\item $e = \sigma$ is a regular expression, and its language is $\Lreg(e) = \{ \sigma \}$.
\item $e = e_1 | e_2$ is a regular expression, and $\Lreg(e) = \Lreg(e_1) \cup \Lreg(e_2)$.
\item $e = e_1 \circ e_2$ is a regular expression, and $\Lreg(e) = \{ x=x_1 \circ x_2 \mid x_1 \in \Lreg(e_1) \text{ and } x_2 \in \Lreg(e_2) \}$.
\end{itemize}

As is popular, we will assume that regular expressions are given as strings over $\Sigma \cup \{ |, \circ \}$, and the length of the regular expression is just the length of this string. 
We will often omit the concatenation notation so that $e_1e_2$ represents $e_1 \circ e_2$ (when clear from context).
Our reductions will produce regular expressions whose language will only contain strings of a fixed length $M$, which will be natural in the context of our Closest Pair problem.
Recall that the Hamming distance of two strings $x,y$ of length $m$ is $H(x,y):=  m - \sum_{i=1}^m (x[i] = y[i])$.

\begin{defn}
[{\sc RegExp Closest Pair Problem}]
Given a set of $n$ strings $Y \subseteq \Sigma^m$ of length $m$ over an alphabet $\Sigma$, and a regular expression $\tilde{x}$ of length $N$ whose language $\Lreg(\tilde{x}) \subseteq \Sigma^m$ contains strings of length $m$, 
the goal is to find a pair of strings $x \in \Lreg(\tilde{x}), y \in Y$ with minimum Hamming distance $H(x,y)$.
\end{defn}

The main result of this section proves a hardness result for this problem which holds even for instances with perfect completeness: even if an exact match exists, it is hard to find any pair with Hamming distance $O(m^{1-\varepsilon})$, for any $\varepsilon>0$.
This is essentially as strong as possible, since a factor $m$ approximation is trivial (return any string).
For the case of binary alphabets, we show that even if an exact match exists (pair of distance $0$), it is hard to find a pair of distance $(1/2-\varepsilon) \cdot m$, for any $\varepsilon>0$.

\begin{thm}
\label{thm:regexp}

Let $\varepsilon>0$ be any constant, and
let $\left(\tilde{x},Y,\Sigma \right)$ be an instance of {\sc RegExp Closest Pair} with
$|Y|=n$ strings and parameters $m,|\Sigma|=n^{o(1)}$, and $N=n^{1+o(1)}$. Then, assuming SETH, $O\left(n^{2-\varepsilon}\right)$-time
algorithms cannot distinguish between:
\begin{description}
\item [{Completeness}] there exist $x \in \Lreg(\tilde{x}),y \in Y$ such that $H(x,y)=0$;
and
\item [{Soundness}] for every $x\in \Lreg(\tilde{x}),y\in Y$, $H(x,y) = m \cdot (1- \frac{1}{ 2^{(\log{n})^{1-o(1)}} })$.
\end{description}

Furthermore, if we restrict to binary alphabets $|\Sigma|=2$, then for all $\delta>0$, assuming SETH, $O\left(n^{2-\varepsilon}\right)$-time
algorithms cannot distinguish between:
\begin{description}
\item [{Completeness}] there exist $x \in \Lreg(\tilde{x}),y \in Y$ such that $H(x,y)=0$;
and
\item [{Soundness}] for every $x\in \Lreg(\tilde{x}),y\in Y$, $H(x,y) = (1/2 - \delta)\cdot m$.
\end{description}

\begin{proof}

We start with the simpler case of large alphabets. 
The reduction essentially translates the {\sc PCP-Vectors} problem into the language of regular expressions.
Given an instance of {\sc PCP-Vectors} as in Theorem~\ref{thm:main} with $N$ vectors and parameters $d_{A},K,|\Sigma|=N^{o\left(1\right)}$, and $d_{B}=L =O(\log{N})$, 
we construct an instance of {\sc RegExp Closest Pair Problem} as follows.
Our alphabet will be the same $\Sigma$, and we will map each vector $b \in B$ into a string $y(b) \in Y$ in the obvious way: for all $\ell \in [L]$ $y(b)[\ell] := b_\ell$.
Note that we have $n=N$ strings in $Y$ and their length is $m=L=O(\log{N})$, and the alphabet has size $n^{o(1)}$.
The vectors in $A$ will be encoded with the regular expression $\tilde{x}$.
First, for all $\ell \in [L], k \in [K]$ we define the subexpression $\widetilde{x_{\ell,k}}$ to be the single letter $a_{\ell,k} \in \Sigma$. 
Then, for all $\ell \in [L]$ we define the subexpression $\widetilde{x_{\ell}}$ as the OR of all the $\widetilde{x_{\ell,k}}$ strings:
$$
\widetilde{x_{\ell}} := \widetilde{x_{\ell,1}} \  | \ \widetilde{x_{\ell,2}} | \ \cdots \ | \ \widetilde{x_{\ell,K}} 
$$
Then, we concatenate all the $\widetilde{x_{\ell}}$ gadgets into the expression 
$$ \widetilde{x(a)}:= \widetilde{x_{1}} \circ \cdots \circ \widetilde{x_{L}} $$
which encodes a vector $a \in A$.
Finally, the final regular expression $\widetilde{x}$ is the OR of all these gadgets:
$$
\widetilde{x} := \widetilde{x(a^{1})} \  | \ \widetilde{{x(a^{2})}} | \ \cdots \ | \ \widetilde{{x(a^{N})}}.
$$
Note that the length of the expression is $N'=O(NLK)=n^{1+o(1)}$.

To see the correctness of the reduction, first observe that for all $a\in A ,b\in B$ we have that:
$$
\min_{x \ \in \ \widetilde{x(a)}} H(x,y(b)) = \sum_{\ell=1}^L \left( 1- \max_{k \in [K]} (a_{\ell,k} = b_\ell) \right)
 =  L - \sum_{\ell=1}^L \bigvee_{k \in [K]} (a_{\ell,k} = b_\ell)
$$
Which follows since the language of $\widetilde{x}$ can be expressed as:
$$
\Lreg(\widetilde{x(a)}) = 
\left\{
x=a_{1,k_1}\cdots a_{L,k_L}
\mid 
k_1,\ldots,k_L \in [K]
\right\}
$$
It follows that for all $ b\in B$:
$$
\min_{x \ \in \ \widetilde{x}} H(x,y(b)) = \min_{a \in A} \left( L - \sum_{\ell=1}^L \bigvee_{k \in [K]} (a_{\ell,k} = b_\ell) \right) = L - \max_{a \in A} L\cdot s(a,b)
$$

\paragraph{Completeness.} if there is a pair $a\in A, b \in B$ with $s(a,b)=1$, then for $y(b) \in B$ we have that $\min_{x \ \in \ \widetilde{x}} H(x,y(b)) = 0$.

\paragraph{Soundness.} if for all pairs $a \in A, b \in B$ $s(a,b) =  \frac{1}{2^{(\log{n})^{1-o(1)}}}$ then for all strings $y \in Y$ and $x \in \widetilde{x}$ we have that $H(x,y) = L \cdot \left(1 - \frac{1}{2^{(\log{n})^{1-o(1)}}} \right)$.
 
 \medskip
 \paragraph{Binary Alphabets.}
To get a reduction to strings over a binary alphabet, we simply replace each letter in $\Sigma$ with a codeword from a code with large Hamming distance between any pair of codewords.
Let $\delta>0$ be an arbitrary small constant, and let $d= \polylog({|\Sigma|})$. 
We consider an error correcting code $e : \Sigma \to \{0,1\}^d$ with constant rate and relative distance $(1/2-\delta)$; i.e. for any distinct $\sigma,\sigma' \in \Sigma$ we have that $H(e(\sigma),e(\sigma')) \geq (1/2-\delta)\cdot d$ (e.g. a random code or the concatenation of Reed Solomon and Hadamard codes~\cite[Chapter 17.5.3]{AB09-book}).
We map each symbol $\sigma \in \Sigma$ in any of our strings $y\in Y$ or in the regular expression $\widetilde{x}$ into the string $e(\sigma)$.
For any two strings $x,y \in \Sigma^L$ let $x',y' \in \{ 0,1 \}^{Ld}$ be the strings after this replacement, and observe that $H(x',y') \geq H(x,y) \cdot (1/2-\delta)\cdot d$, and moreover, if $H(x,y)=0$, then $H(x',y')=0$.
 The completeness and soundness follow, and note that the lengths of the strings and the expression grow by a negligible $n^{o(1)}$ factor.
\end{proof}

\end{thm}

\section{Product metric diameter}

\begin{defn}
[{\sc Product-Metric Diameter}] The input to this problem consists of a set of vectors $X\subset\mathbb{R}^{d_{2}\times d_{\infty}}$.

The goal is to find two vectors $x,y\in X$ that maximize
\begin{gather}
\Delta_{2,\infty}(x,y) \triangleq 
	\sqrt{\sum_{i =1}^{d_2} \left( 
		\max_{j = 2}^{d_{\infty}} \left\{
				\Big| x_{i,j} - y_{i,j} \Big|
				\right\} \right)^2}.
\end{gather}
\end{defn}

\begin{thm}
\label{thm:product}
Assuming SETH, there are no $(2-\delta)$-approximation algorithms for {\sc Product-Metric Diameter} in time $O\left(N^{2-\varepsilon}\right)$, for any constants $\varepsilon,\delta>0$.
\end{thm}
\begin{proof}
We reduce from {\sc PCP-Vectors} over alphabet $\Sigma = N^{o(1)}$.

For every vector $a \in A \subset \Sigma^{L \times K}$, we construct a binary vector $x(a) \in \{0,1\}^{L \times \Sigma}$ by setting:
\begin{gather*}
x(a) \triangleq \bigcirc_{\ell \in L} \bigcirc_{\sigma \in \Sigma} \bigvee_{k \in K}[a_{\ell,k} = \sigma].
\end{gather*}
Similarly, for each $b \in B \subset \Sigma^{L}$, construct a binary vector $y(b) \in \{-1,0\}^{L \times \Sigma}$ by setting:
\begin{gather*}
y(b) \triangleq -\bigcirc_{\ell \in L} \bigcirc_{\sigma \in \Sigma} [b_{\ell} = \sigma].
\end{gather*}

For any pair $x(a), y(b)$ and $\ell \in L$, we have that the $\infty$-norm distance between $x(a)_{\ell,\cdot}$ and $y(b)_{\ell, \cdot}$ is $2$ if there is some $k$ such that $a_{\ell,k} = b_{\ell}$, and $1$ otherwise.
Therefore, summing over $\ell \in L$ we have that
$$
\Delta_{2,\infty}\big(x(a),y(b)\big) = (1+s(a,b))\cdot \ell.
$$
In particular, if $(A,B)$ is a yes case of {\sc PCP-Vectors}, there are $x(a^*),y(b^*)$ such that $\Delta_{2,\infty}\big(x(a^*),y(b^*)\big) = 2\ell$; given a no instance, $\Delta_{2,\infty}\big(x(a),y(b)\big) = (1+o(1))\ell$ for any $x(a),y(b)$. 

Finally, observe that any two $x$-vectors have $\Delta_{2,\infty}$-distance at most $\ell$ since all the entries have the same sign; similarly for any two $y$-vectors.

\end{proof}

\section{Erratum}
\label{sec:erratum}

In a previous version of this paper, we erroneously claimed hardness of approximation results on symmetric, or ``monochromatic'' variants of {\sc PCP-Vectors} and {\sc Max-IP}, as well as exact hardness for Closest Pair problems in Hamming, Manhattan, and Euclidean metrics. We later discovered a mistake in that part of the proof. Unfortunately, it was too late to update the extended abstract that will appear in the proceedings of FOCS 2017~\cite{ARW17-proceedings}.

\fi

\section*{Acknowledgements}
We thank Karthik C.S., Alessandro Chiesa, S{\o}ren Dahlgaard, Piotr Indyk, Rasmus Pagh, Ilya Razenshteyn, Omer Reingold, Nick Spooner, Virginia Vassilevska Williams, Ameya Velingker, and anonymous reviewers for helpful discussions and suggestions.

This work was done in part at the Simons Institute for the Theory of Computing.
We are also grateful to the organizers of Dagstuhl Seminar 16451 for a special collaboration opportunity.

\ifFULL

\else
A.A. was supported by the grants of Virginia Vassilevska Williams: NSF Grants CCF-1417238, CCF-1528078 and CCF-1514339, and BSF Grant BSF:2012338.
A.R. was supported by Microsoft Research PhD Fellowship, NSF grant CCF-1408635 and by Templeton Foundation grant 3966. 
R.W. was supported by an NSF CAREER grant (CCF-1741615).
\fi

\ifFULL
\bibliographystyle{amsalpha}
\else
\bibliographystyle{IEEEtran}
\fi
\bibliography{bib}

\newcommand{\etalchar}[1]{$^{#1}$}
\providecommand{\bysame}{\leavevmode\hbox to3em{\hrulefill}\thinspace}
\providecommand{\MR}{\relax\ifhmode\unskip\space\fi MR }
\providecommand{\MRhref}[2]{%
  \href{http://www.ams.org/mathscinet-getitem?mr=#1}{#2}
}
\providecommand{\href}[2]{#2}
\begin{thebibliography}{AHWW16}

\bibitem[AAK10]{AAK10}
Parag Agrawal, Arvind Arasu, and Raghav Kaushik, \emph{On indexing
  error-tolerant set containment}, Proc.\ of the 2010 ACM SIGMOD International
  Conference on Management of data, ACM, 2010, pp.~927--938.

\bibitem[AB09]{AB09-book}
Sanjeev Arora and Boaz Barak, \emph{Computational complexity - {A} modern
  approach}, Cambridge University Press, 2009.

\bibitem[AB17]{AB17}
Amir Abboud and Arturs Backurs, \emph{Towards hardness of approximation for
  polynomial time problems}, ITCS, to appear, 2017.

\bibitem[ABV15]{ABV15a}
Amir Abboud, Arturs Backurs, and Virginia {Vassilevska Williams}, \emph{Tight
  hardness results for {LCS} and other sequence similarity measures}, Proc. of
  the 56th FOCS, 2015, pp.~59--78.

\bibitem[ACIM99]{ACIM99}
Donald Aingworth, Chandra Chekuri, Piotr Indyk, and Rajeev Motwani, \emph{Fast
  estimation of diameter and shortest paths (without matrix multiplication)},
  {SIAM} J. Comput. \textbf{28} (1999), no.~4, 1167--1181.

\bibitem[ACP08]{ACP08}
Alexandr Andoni, Dorian Croitoru, and Mihai Patrascu, \emph{Hardness of nearest
  neighbor under l-infinity}, Proc.\ of the 49th FOCS, 2008, pp.~424--433.

\bibitem[ADG{\etalchar{+}}03]{ADG+03}
Alexandr Andoni, Michel Deza, Anupam Gupta, Piotr Indyk, and Sofya
  Raskhodnikova, \emph{Lower bounds for embedding edit distance into normed
  spaces}, Proc.\ of the 14th SODA, 2003, pp.~523--526.

\bibitem[AHWW16]{AHVW16}
Amir Abboud, Thomas~Dueholm Hansen, Virginia~Vassilevska Williams, and Ryan
  Williams, \emph{Simulating branching programs with edit distance and friends:
  or: a polylog shaved is a lower bound made}, Proc.\ of the 48th STOC, 2016,
  pp.~375--388.

\bibitem[AI06]{AI06}
Alexandr Andoni and Piotr Indyk, \emph{Near-optimal hashing algorithms for
  approximate nearest neighbor in high dimensions}, Proc.\ of the 47th FOCS,
  IEEE, 2006, pp.~459--468.

\bibitem[AIK09]{AIK09}
Alexandr Andoni, Piotr Indyk, and Robert Krauthgamer, \emph{Overcoming the
  \emph{l}\({}_{\mbox{1}}\) non-embeddability barrier: algorithms for product
  metrics}, Proc.\ of the 20th SODA, 2009, pp.~865--874.

\bibitem[AIL{\etalchar{+}}15]{AILRS15}
Alexandr Andoni, Piotr Indyk, Thijs Laarhoven, Ilya Razenshteyn, and Ludwig
  Schmidt, \emph{Practical and optimal lsh for angular distance}, Advances in
  Neural Information Processing Systems, 2015, pp.~1225--1233.

\bibitem[AINR14]{AINR14}
Alexandr Andoni, Piotr Indyk, Huy~L Nguyen, and Ilya Razenshteyn, \emph{Beyond
  locality-sensitive hashing}, Proc.\ of the 25th SODA, SIAM, 2014,
  pp.~1018--1028.

\bibitem[AIP06]{AIP06}
Alexandr Andoni, Piotr Indyk, and Mihai Patrascu, \emph{On the optimality of
  the dimensionality reduction method}, Proc.\ of the 47th FOCS, IEEE, 2006,
  pp.~449--458.

\bibitem[AJP10]{AJP10}
Alexandr Andoni, T.~S. Jayram, and Mihai Patrascu, \emph{Lower bounds for edit
  distance and product metrics via poincar{\'{e}}-type inequalities}, Proc.\ of
  the 21st SODA, 2010, pp.~184--192.

\bibitem[AK07]{AK07}
Alexandr Andoni and Robert Krauthgamer, \emph{The computational hardness of
  estimating edit distance [extended abstract]}, Proc.\ of the 48th FOCS, 2007,
  pp.~724--734.

\bibitem[AKO10]{andoni2010polylogarithmic}
Alexandr Andoni, Robert Krauthgamer, and Krzysztof Onak, \emph{Polylogarithmic
  approximation for edit distance and the asymmetric query complexity}, FOCS,
  2010, pp.~377--386.

\bibitem[ALM{\etalchar{+}}98]{ALMSS98-PCP}
Sanjeev Arora, Carsten Lund, Rajeev Motwani, Madhu Sudan, and Mario Szegedy,
  \emph{Proof verification and the hardness of approximation problems}, J.
  {ACM} \textbf{45} (1998), no.~3, 501--555.

\bibitem[ALRW17]{ALRW16}
Alexandr Andoni, Thijs Laarhoven, Ilya~P. Razenshteyn, and Erik Waingarten,
  \emph{Optimal hashing-based time-space trade-offs for approximate near
  neighbors}, Proc.\ of the 28th SODA, 2017, pp.~47--66.

\bibitem[AO12]{andoni2012approximating}
Alexandr Andoni and Krzysztof Onak, \emph{Approximating edit distance in
  near-linear time}, SIAM Journal on Computing \textbf{41} (2012), no.~6,
  1635--1648.

\bibitem[APRS16]{pods16}
Thomas~Dybdahl Ahle, Rasmus Pagh, Ilya Razenshteyn, and Francesco Silvestri,
  \emph{On the complexity of inner product similarity join}, Proc.\ of the 35th
  ACM SIGMOD-SIGACT-SIGAI Symposium on Principles of Database Systems, ACM,
  2016, pp.~151--164.

\bibitem[AR15]{AR15}
Alexandr Andoni and Ilya Razenshteyn, \emph{Optimal data-dependent hashing for
  approximate near neighbors}, Proc.\ of the Forty-Seventh Annual ACM on
  Symposium on Theory of Computing, ACM, 2015, pp.~793--801.

\bibitem[ARW17]{ARW17-proceedings}
Amir Abboud, Aviad Rubinstein, and Ryan Williams, \emph{{Distributed PCP
  Theorems for Hardness of Approximation in P}}, FOCS, to appear, 2017.

\bibitem[AS98]{AS98-PCP}
Sanjeev Arora and Shmuel Safra, \emph{Probabilistic checking of proofs: {A} new
  characterization of {NP}}, J. {ACM} \textbf{45} (1998), no.~1, 70--122.

\bibitem[AV14]{AV14}
Amir Abboud and Virginia {Vassilevska Williams}, \emph{Popular conjectures
  imply strong lower bounds for dynamic problems}, Proc. of the 55th FOCS,
  2014, pp.~434--443.

\bibitem[AV15]{AV15}
Amirali Abdullah and Suresh Venkatasubramanian, \emph{A directed isoperimetric
  inequality with application to bregman near neighbor lower bounds}, Proc.\ of
  the 47th STOC, ACM, 2015, pp.~509--518.

\bibitem[AW09]{AW09-algebrization}
Scott Aaronson and Avi Wigderson, \emph{Algebrization: {A} new barrier in
  complexity theory}, {TOCT} \textbf{1} (2009), no.~1, 2:1--2:54.

\bibitem[AW15]{AW15}
Josh Alman and Ryan Williams, \emph{Probabilistic polynomials and hamming
  nearest neighbors}, Proc.\ of the 56th FOCS, IEEE, 2015, pp.~136--150.

\bibitem[AWW14]{AVW14}
Amir Abboud, Virginia~Vassilevska Williams, and Oren Weimann,
  \emph{Consequences of faster alignment of sequences}, Proc.\ of the 41st
  ICALP, 2014, pp.~39--51.

\bibitem[BCG{\etalchar{+}}16]{BCGRS16-IOP}
Eli Ben{-}Sasson, Alessandro Chiesa, Ariel Gabizon, Michael Riabzev, and
  Nicholas Spooner, \emph{Short interactive oracle proofs with constant query
  complexity, via composition and sumcheck}, {IACR} Cryptology ePrint Archive
  \textbf{2016} (2016), 324.

\bibitem[BCS16]{BCS16-IOP}
Eli Ben{-}Sasson, Alessandro Chiesa, and Nicholas Spooner, \emph{Interactive
  oracle proofs}, Theory of Cryptography - 14th International Conference,
  {TCC}, 2016, pp.~31--60.

\bibitem[BCW98]{BCW98-quantum_communication}
Harry Buhrman, Richard Cleve, and Avi Wigderson, \emph{Quantum vs. classical
  communication and computation}, Proc.\ of the Thirtieth Annual {ACM}
  Symposium on the Theory of Computing, 1998, pp.~63--68.

\bibitem[BES06]{batu2006oblivious}
Tu{\u{g}}kan Batu, Funda Ergun, and Cenk Sahinalp, \emph{Oblivious string
  embeddings and edit distance approximations}, Proc.\ of the seventeenth
  annual ACM-SIAM symposium on Discrete algorithm, Society for Industrial and
  Applied Mathematics, 2006, pp.~792--801.

\bibitem[BGL16]{BGL16}
Karl Bringmann, Allan Gr{\o}nlund, and Kasper~Green Larsen, \emph{A dichotomy
  for regular expression membership testing}, CoRR \textbf{abs/1611.00918}
  (2016).

\bibitem[BGMZ97]{Broder+97}
Andrei~Z Broder, Steven~C Glassman, Mark~S Manasse, and Geoffrey Zweig,
  \emph{Syntactic clustering of the web}, Computer Networks and ISDN Systems
  \textbf{29} (1997), no.~8-13, 1157--1166.

\bibitem[BGS15]{baswana2015fully}
Surender Baswana, Manoj Gupta, and Sandeep Sen, \emph{Fully dynamic maximal
  matching in o($\backslash$logn) update time}, SIAM Journal on Computing
  \textbf{44} (2015), no.~1, 88--113.

\bibitem[BHI14]{bhattacharya2014deterministic}
Sayan Bhattacharya, Monika Henzinger, and Giuseppe~F Italiano,
  \emph{Deterministic fully dynamic data structures for vertex cover and
  matching}, Proc.\ of the Twenty-Sixth Annual ACM-SIAM Symposium on Discrete
  Algorithms, SIAM, 2014, pp.~785--804.

\bibitem[BHI16]{bhattacharya2016design}
\bysame, \emph{Design of dynamic algorithms via primal-dual method}, arXiv
  preprint arXiv:1604.05337 (2016).

\bibitem[BHN16]{bhattacharya2016new}
Sayan Bhattacharya, Monika Henzinger, and Danupon Nanongkai, \emph{New
  deterministic approximation algorithms for fully dynamic matching}, Proc.\ of
  the 48th Annual ACM SIGACT Symposium on Theory of Computing, ACM, 2016,
  pp.~398--411.

\bibitem[BI15]{BI15}
Arturs Backurs and Piotr Indyk, \emph{{Edit Distance Cannot Be Computed in
  Strongly Subquadratic Time (unless SETH is false)}}, Proc. of the 47th Annual
  {ACM} {SIGACT} Symposium on Theory of Computing ({STOC}), 2015, pp.~51--58.

\bibitem[BI16]{BI16}
\bysame, \emph{Which regular expression patterns are hard to match?}, Proc. of
  the 57th Annual {IEEE} Symposium on Foundations of Computer Science ({FOCS}),
  2016, pp.~457--466.

\bibitem[BJKS04]{BJKS04-communication}
Ziv Bar{-}Yossef, T.~S. Jayram, Ravi Kumar, and D.~Sivakumar, \emph{An
  information statistics approach to data stream and communication complexity},
  J. Comput. Syst. Sci. \textbf{68} (2004), no.~4, 702--732.

\bibitem[BK15]{BK15}
Karl Bringmann and Marvin Kunnemann, \emph{Quadratic conditional lower bounds
  for string problems and dynamic time warping}, Proc. of the 56th Annual
  {IEEE} Symposium on Foundations of Computer Science ({FOCS}), 2015,
  pp.~79--97.

\bibitem[BKK{\etalchar{+}}16]{BKKMS16-linear_PCP}
Eli Ben{-}Sasson, Yohay Kaplan, Swastik Kopparty, Or~Meir, and Henning
  Stichtenoth, \emph{Constant rate pcps for circuit-sat with sublinear query
  complexity}, J. {ACM} \textbf{63} (2016), no.~4, 32:1--32:57.

\bibitem[BOR04]{BOR04-diameter}
Allan Borodin, Rafail Ostrovsky, and Yuval Rabani, \emph{Subquadratic
  approximation algorithms for clustering problems in high dimensional spaces},
  Machine Learning \textbf{56} (2004), no.~1-3, 153--167.

\bibitem[BR13]{belazzougui2013approximate}
Djamal Belazzougui and Mathieu Raffinot, \emph{Approximate regular expression
  matching with multi-strings}, Journal of Discrete Algorithms \textbf{18}
  (2013), 14--21.

\bibitem[Bri14]{Bring14}
Karl Bringmann, \emph{Why walking the dog takes time: {F}rechet distance has no
  strongly subquadratic algorithms unless {SETH} fails}, Proc. of the 55th
  Annual {IEEE} Symposium on Foundations of Computer Science ({FOCS}), 2014,
  pp.~661--670.

\bibitem[Bro97]{Broder97}
Andrei~Z Broder, \emph{On the resemblance and containment of documents},
  Compression and Complexity of Sequences 1997. Proceedings, IEEE, 1997,
  pp.~21--29.

\bibitem[BRSV17]{BRSV17-average_case}
Marshall Ball, Alon Rosen, Manuel Sabin, and Prashant~Nalini Vasudevan,
  \emph{Average-case fine-grained hardness}, {IACR} Cryptology ePrint Archive
  \textbf{2017} (2017), 202.

\bibitem[BS15]{bernstein2015fully}
Aaron Bernstein and Cliff Stein, \emph{Fully dynamic matching in bipartite
  graphs}, arXiv preprint arXiv:1506.07076 (2015).

\bibitem[BS16]{bernstein2016faster}
\bysame, \emph{Faster fully dynamic matchings with small approximation ratios},
  Proc.\ of the Twenty-Seventh Annual ACM-SIAM Symposium on Discrete
  Algorithms, Society for Industrial and Applied Mathematics, 2016,
  pp.~692--711.

\bibitem[BT09]{BT09}
Philip Bille and Mikkel Thorup, \emph{Faster regular expression matching},
  International Colloquium on Automata, Languages, and Programming, Springer,
  2009, pp.~171--182.

\bibitem[BYJKK04]{bar2004approximating}
Ziv Bar-Yossef, TS~Jayram, Robert Krauthgamer, and Ravi Kumar,
  \emph{Approximating edit distance efficiently}, Foundations of Computer
  Science, 2004. Proceedings. 45th Annual IEEE Symposium on, IEEE, 2004,
  pp.~550--559.

\bibitem[CDL{\etalchar{+}}16]{Cygan+16}
Marek Cygan, Holger Dell, Daniel Lokshtanov, D{\'a}niel Marx, Jesper Nederlof,
  Yoshio Okamoto, Ramamohan Paturi, Saket Saurabh, and Magnus Wahlstr{\"o}m,
  \emph{On problems as hard as {CNF-SAT}}, ACM Transactions on Algorithms
  \textbf{12} (2016), no.~3, 41.

\bibitem[Chr17]{Chris17}
Tobias Christiani, \emph{A framework for similarity search with space-time
  tradeoffs using locality-sensitive filtering}, Proc.\ of the Twenty-Eighth
  Annual ACM-SIAM Symposium on Discrete Algorithms, SIAM, 2017, pp.~31--46.

\bibitem[CIP06]{CIP06}
Chris Calabro, Russell Impagliazzo, and Ramamohan Paturi, \emph{A duality
  between clause width and clause density for {SAT}}, Proc. of 21st Conference
  on Computational Complexity ({CCC}), 2006, pp.~252--260.

\bibitem[CLR{\etalchar{+}}14]{CLRSTW14}
Shiri Chechik, Daniel~H. Larkin, Liam Roditty, Grant Schoenebeck, Robert~Endre
  Tarjan, and Virginia~Vassilevska Williams, \emph{Better approximation
  algorithms for the graph diameter}, Proc.\ of the Twenty-Fifth Annual
  {ACM-SIAM} Symposium on Discrete Algorithms, 2014, pp.~1041--1052.

\bibitem[CM05]{CM05-cascaded}
Graham Cormode and S.~Muthukrishnan, \emph{Space efficient mining of multigraph
  streams}, Proc.\ of the Twenty-fourth {ACM} {SIGACT-SIGMOD-SIGART} Symposium
  on Principles of Database Systems, 2005, pp.~271--282.

\bibitem[CMS01]{CMS01}
Graham Cormode, S~Muthukrishnan, and S{\"u}leyman~Cenk Sahinalp,
  \emph{Permutation editing and matching via embeddings}, International
  Colloquium on Automata, Languages, and Programming, Springer, 2001,
  pp.~481--492.

\bibitem[CP16]{CP16}
Tobias Christiani and Rasmus Pagh, \emph{Set similarity search beyond minhash},
  arXiv preprint arXiv:1612.07710 (2016).

\bibitem[Dah16]{Dahlgaard16}
S{\o}ren Dahlgaard, \emph{On the hardness of partially dynamic graph problems
  and connections to diameter}, Proc.\ of the 43rd {ICALP}, 2016,
  pp.~48:1--48:14.

\bibitem[Din07]{Dinur07-PCP}
Irit Dinur, \emph{The {PCP} theorem by gap amplification}, J. {ACM} \textbf{54}
  (2007), no.~3, 12.

\bibitem[Din16]{Dinur16-gap_ETH}
\bysame, \emph{Mildly exponential reduction from gap 3sat to polynomial-gap
  label-cover}, Electronic Colloquium on Computational Complexity {(ECCC)}
  \textbf{23} (2016), 128.

\bibitem[DKL16]{DSL16-bichromatic}
Roee David, {Karthik {C. S.}}, and Bundit Laekhanukit, \emph{The curse of
  medium dimension for geometric problems in almost every norm}, CoRR
  \textbf{abs/1608.03245} (2016).

\bibitem[Dru10]{Drucker10_thesis}
Andrew Drucker, \emph{{PCPs for Arthur-Merlin Games and Communication
  Protocols}}, Master's thesis, Massachusetts Institute of Technology, 2010.

\bibitem[EK89]{EK89-diameter_sqrt3}
{\"{O}}mer Egecioglu and Bahman Kalantari, \emph{Approximating the diameter of
  a set of points in the euclidean space}, Inf. Process. Lett. \textbf{32}
  (1989), no.~4, 205--211.

\bibitem[FP02]{FP02-diameter}
Daniele~V. Finocchiaro and Marco Pellegrini, \emph{On computing the diameter of
  a point set in high dimensional euclidean space}, Theor. Comput. Sci.
  \textbf{287} (2002), no.~2, 501--514.

\bibitem[GG10]{GG10}
Ashish Goel and Pankaj Gupta, \emph{Small subset queries and bloom filters
  using ternary associative memories, with applications}, ACM SIGMETRICS
  Performance Evaluation Review \textbf{38} (2010), no.~1, 143--154.

\bibitem[GIKW17]{GIKW17}
Jiawei Gao, Russell Impagliazzo, Antonina Kolokolova, and R.~Ryan Williams,
  \emph{Completeness for first-order properties on sparse structures with
  algorithmic applications}, Proc. of the 28th Annual {ACM-SIAM} Symposium on
  Discrete Algorithms ({SODA}), 2017, pp.~2162--2181.

\bibitem[GIMS10]{GIMS10-zkPCP}
Vipul Goyal, Yuval Ishai, Mohammad Mahmoody, and Amit Sahai, \emph{Interactive
  locking, zero-knowledge pcps, and unconditional cryptography}, Advances in
  Cryptology - {CRYPTO}, 30th Annual Cryptology Conference,, 2010,
  pp.~173--190.

\bibitem[GIV01]{GIV01-diameter}
Ashish Goel, Piotr Indyk, and Kasturi~R. Varadarajan, \emph{Reductions among
  high dimensional proximity problems}, Proc.\ of the Twelfth Annual Symposium
  on Discrete Algorithms, 2001, pp.~769--778.

\bibitem[GKR15]{GKR15-delegation_muggles}
Shafi Goldwasser, Yael~Tauman Kalai, and Guy~N. Rothblum, \emph{Delegating
  computation: Interactive proofs for muggles}, J. {ACM} \textbf{62} (2015),
  no.~4, 27:1--27:64.

\bibitem[GO95]{GO95}
Anka Gajentaan and Mark~H. Overmars, \emph{On a class of o(n2) problems in
  computational geometry}, Comput. Geom. \textbf{5} (1995), 165--185.

\bibitem[GP13]{gupta2013fully}
Manoj Gupta and Richard Peng, \emph{Fully dynamic (1+ e)-approximate
  matchings}, Proc.\ of the 54th FOCS, IEEE, 2013, pp.~548--557.

\bibitem[H{\aa}s01]{Hastad01}
Johan H{\aa}stad, \emph{Some optimal inapproximability results}, J. {ACM}
  \textbf{48} (2001), no.~4, 798--859.

\bibitem[HKNS15]{HKNS15}
Monika Henzinger, Sebastian Krinninger, Danupon Nanongkai, and Thatchaphol
  Saranurak, \emph{Unifying and strengthening hardness for dynamic problems via
  the online matrix-vector multiplication conjecture}, Proc.\ of the
  Forty-Seventh Annual {ACM} on Symposium on Theory of Computing, {STOC}, 2015,
  pp.~21--30.

\bibitem[IM98]{IM98}
Piotr Indyk and Rajeev Motwani, \emph{Approximate nearest neighbors: towards
  removing the curse of dimensionality}, Proc.\ of the thirtieth annual ACM
  symposium on Theory of computing, ACM, 1998, pp.~604--613.

\bibitem[Ind98]{Indyk98-product}
Piotr Indyk, \emph{On approximate nearest neighbors in non-euclidean spaces},
  39th Annual Symposium on Foundations of Computer Science, {FOCS}, 1998,
  pp.~148--155.

\bibitem[Ind99]{Indyk99-diameter_lb}
\bysame, \emph{A sublinear time approximation scheme for clustering in metric
  spaces}, 40th Annual Symposium on Foundations of Computer Science, {FOCS},
  1999, pp.~154--159.

\bibitem[Ind00]{Indyk00-diameter}
\bysame, \emph{Dimensionality reduction techniques for proximity problems},
  Proc.\ of the Eleventh Annual {ACM-SIAM} Symposium on Discrete Algorithms,
  2000, pp.~371--378.

\bibitem[Ind02]{Indyk02-product}
\bysame, \emph{Approximate nearest neighbor algorithms for frechet distance via
  product metrics}, Proc.\ of the 18th Annual Symposium on Computational
  Geometry,, 2002, pp.~102--106.

\bibitem[Ind03]{Indyk03-diameter}
\bysame, \emph{Better algorithms for high-dimensional proximity problems via
  asymmetric embeddings}, Proc.\ of the Fourteenth Annual {ACM-SIAM} Symposium
  on Discrete Algorithms, 2003, pp.~539--545.

\bibitem[Ind04]{Indyk04-product}
\bysame, \emph{Approximate nearest neighbor under edit distance via product
  metrics}, Proc.\ of the Fifteenth Annual {ACM-SIAM} Symposium on Discrete
  Algorithms, {SODA}, 2004, pp.~646--650.

\bibitem[IP01]{IP01-SETH}
Russell Impagliazzo and Ramamohan Paturi, \emph{On the complexity of k-sat}, J.
  Comput. Syst. Sci. \textbf{62} (2001), no.~2, 367--375.

\bibitem[IPZ01]{IPZ01-ETH}
Russell Impagliazzo, Ramamohan Paturi, and Francis Zane, \emph{Which problems
  have strongly exponential complexity?}, J. Comput. Syst. Sci. \textbf{63}
  (2001), no.~4, 512--530.

\bibitem[JW09]{JW09-cascaded}
T.~S. Jayram and David~P. Woodruff, \emph{The data stream space complexity of
  cascaded norms}, 50th Annual {IEEE} Symposium on Foundations of Computer
  Science, {FOCS}, 2009, pp.~765--774.

\bibitem[KKK16]{KarKK16}
Matti Karppa, Petteri Kaski, and Jukka Kohonen, \emph{A faster subquadratic
  algorithm for finding outlier correlations}, Proc.\ of the Twenty-Seventh
  Annual ACM-SIAM Symposium on Discrete Algorithms, Society for Industrial and
  Applied Mathematics, 2016, pp.~1288--1305.

\bibitem[KKN95]{KKN95-disjointness}
Mauricio Karchmer, Eyal Kushilevitz, and Noam Nisan, \emph{Fractional covers
  and communication complexity}, {SIAM} J. Discrete Math. \textbf{8} (1995),
  no.~1, 76--92.

\bibitem[KM95]{knight1995approximate}
James~R Knight and Eugene~W Myers, \emph{Approximate regular expression pattern
  matching with concave gap penalties}, Algorithmica \textbf{14} (1995), no.~1,
  85--121.

\bibitem[KN05]{KN05}
Subhash Khot and Assaf Naor, \emph{Nonembeddability theorems via fourier
  analysis}, Proc.\ of the 46th FOCS, IEEE, 2005, pp.~101--110.

\bibitem[KP12]{KP12}
Michael Kapralov and Rina Panigrahy, \emph{Nns lower bounds via metric
  expansion for $l_{\infty}$ and emd}, International Colloquium on Automata,
  Languages, and Programming, Springer, 2012, pp.~545--556.

\bibitem[KPP16]{KPP16}
Tsvi Kopelowitz, Seth Pettie, and Ely Porat, \emph{Higher lower bounds from the
  3sum conjecture}, Proc.\ of the Twenty-Seventh Annual ACM-SIAM Symposium on
  Discrete Algorithms, Society for Industrial and Applied Mathematics, 2016,
  pp.~1272--1287.

\bibitem[KR08]{KR08-interactive_pcp}
Yael~Tauman Kalai and Ran Raz, \emph{Interactive {PCP}}, Automata, Languages
  and Programming, 35th International Colloquium, {ICALP}, 2008, pp.~536--547.

\bibitem[KS92]{KS92-communication}
Bala Kalyanasundaram and Georg Schnitger, \emph{The probabilistic communication
  complexity of set intersection}, {SIAM} J. Discrete Math. \textbf{5} (1992),
  no.~4, 545--557.

\bibitem[LFKN92]{LFKN92}
Carsten Lund, Lance Fortnow, Howard~J. Karloff, and Noam Nisan, \emph{Algebraic
  methods for interactive proof systems}, J. {ACM} \textbf{39} (1992), no.~4,
  859--868.

\bibitem[LMS98]{landau1998incremental}
Gad~M Landau, Eugene~W Myers, and Jeanette~P Schmidt, \emph{Incremental string
  comparison}, SIAM Journal on Computing \textbf{27} (1998), no.~2, 557--582.

\bibitem[LMS11]{LMS11}
Daniel Lokshtanov, D{\'{a}}niel Marx, and Saket Saurabh, \emph{Lower bounds
  based on the exponential time hypothesis}, Bulletin of the {EATCS}
  \textbf{105} (2011), 41--72.

\bibitem[LW17]{LW17}
Kasper~Green Larsen and R.~Ryan Williams, \emph{Faster online matrix-vector
  multiplication}, Proc.\ of the Twenty-Eighth Annual {ACM-SIAM} Symposium on
  Discrete Algorithms, {SODA}, 2017, pp.~2182--2189.

\bibitem[MGM03]{MG03}
Sergey Melnik and Hector Garcia-Molina, \emph{Adaptive algorithms for set
  containment joins}, ACM Transactions on Database Systems (TODS) \textbf{28}
  (2003), no.~1, 56--99.

\bibitem[MM89]{myers1989approximate}
Eugene~W Myers and Webb Miller, \emph{Approximate matching of regular
  expressions}, Bulletin of mathematical biology \textbf{51} (1989), no.~1,
  5--37.

\bibitem[MNP07]{MNP07}
Rajeev Motwani, Assaf Naor, and Rina Panigrahy, \emph{Lower bounds on locality
  sensitive hashing}, SIAM Journal on Discrete Mathematics \textbf{21} (2007),
  no.~4, 930--935.

\bibitem[MOG98]{myers1998reporting}
Eugene Myers, Paulo Oliva, and Katia Guimar{\~a}es, \emph{Reporting exact and
  approximate regular expression matches}, Combinatorial pattern matching,
  Springer, 1998, pp.~91--103.

\bibitem[MPS16]{MPS16}
Daniel Moeller, Ramamohan Paturi, and Stefan Schneider, \emph{Subquadratic
  algorithms for succinct stable matching}, International Computer Science
  Symposium in Russia, Springer, 2016, pp.~294--308.

\bibitem[Mye92]{Mye92}
Gene Myers, \emph{A four russians algorithm for regular expression pattern
  matching}, Journal of the ACM (JACM) \textbf{39} (1992), no.~2, 432--448.

\bibitem[Nav04]{navarro2004approximate}
Gonzalo Navarro, \emph{Approximate regular expression searching with arbitrary
  integer weights.}, Nord. J. Comput. \textbf{11} (2004), no.~4, 356--373.

\bibitem[NS15]{NS15}
Behnam Neyshabur and Nathan Srebro, \emph{On symmetric and asymmetric lshs for
  inner product search}, Proc.\ of the 32nd International Conference on Machine
  Learning, {ICML}, 2015, pp.~1926--1934.

\bibitem[NS16]{neiman2016simple}
Ofer Neiman and Shay Solomon, \emph{Simple deterministic algorithms for fully
  dynamic maximal matching}, ACM Transactions on Algorithms (TALG) \textbf{12}
  (2016), no.~1, 7.

\bibitem[OR07]{ostrovsky2007low}
Rafail Ostrovsky and Yuval Rabani, \emph{Low distortion embeddings for edit
  distance}, Journal of the ACM (JACM) \textbf{54} (2007), no.~5, 23.

\bibitem[OWZ14]{OWZ14}
Ryan O'Donnell, Yi~Wu, and Yuan Zhou, \emph{Optimal lower bounds for
  locality-sensitive hashing (except when q is tiny)}, ACM Transactions on
  Computation Theory (TOCT) \textbf{6} (2014), no.~1, 5.

\bibitem[PS16]{peleg2016dynamic}
David Peleg and Shay Solomon, \emph{Dynamic (1+ $\varepsilon$)-approximate
  matchings: a density-sensitive approach}, Proc.\ of the Twenty-Seventh Annual
  ACM-SIAM Symposium on Discrete Algorithms, Society for Industrial and Applied
  Mathematics, 2016, pp.~712--729.

\bibitem[PTW08]{PTW08}
Rina Panigrahy, Kunal Talwar, and Udi Wieder, \emph{A geometric approach to
  lower bounds for approximate near-neighbor search and partial match}, Proc.\
  of the 49th FOCS, IEEE, 2008, pp.~414--423.

\bibitem[PTW10]{PTW10}
\bysame, \emph{Lower bounds on near neighbor search via metric expansion},
  Proc.\ of the 51st FOCS, IEEE, 2010, pp.~805--814.

\bibitem[P\v10]{Pat10}
Mihai P\v{a}tra\c{s}cu, \emph{Towards polynomial lower bounds for dynamic
  problems}, Proc. of the 42nd Annual {ACM} {S}ymposium on {T}heory {O}f
  {C}omputing ({STOC}), 2010, pp.~603--610.

\bibitem[PW10]{PW10-faster_SAT}
Mihai Patrascu and Ryan Williams, \emph{On the possibility of faster {SAT}
  algorithms}, Proc.\ of the Twenty-First Annual {ACM-SIAM} Symposium on
  Discrete Algorithms, {SODA}, 2010, pp.~1065--1075.

\bibitem[Raz92]{Razborov92-disjointness}
Alexander~A. Razborov, \emph{On the distributional complexity of disjointness},
  Theor. Comput. Sci. \textbf{106} (1992), no.~2, 385--390.

\bibitem[Raz98]{Raz98-parallel_repetition}
Ran Raz, \emph{A parallel repetition theorem}, {SIAM} J. Comput. \textbf{27}
  (1998), no.~3, 763--803.

\bibitem[Rei17]{Reingold17-communication}
Omer Reingold, March 2017, Private Communication.

\bibitem[RG12]{RG12}
Parikshit Ram and Alexander~G Gray, \emph{Maximum inner-product search using
  cone trees}, Proc.\ of the 18th ACM SIGKDD international conference on
  Knowledge discovery and data mining, ACM, 2012, pp.~931--939.

\bibitem[Riv74]{Rivest74}
Ronald~Linn Rivest, \emph{Analysis of associative retrieval algorithms.}, Ph.D.
  thesis, Stanford University, Stanford, CA, USA, 1974, AAI7420230.

\bibitem[Riv76]{Rivest76}
Ronald~L Rivest, \emph{Partial-match retrieval algorithms}, SIAM Journal on
  Computing \textbf{5} (1976), no.~1, 19--50.

\bibitem[RPNK00]{Rama+00}
Karthikeyan Ramasamy, Jignesh~M Patel, Jeffrey~F Naughton, and Raghav Kaushik,
  \emph{Set containment joins: The good, the bad and the ugly.}, VLDB, 2000,
  pp.~351--362.

\bibitem[RR{\etalchar{+}}07]{RR07}
Ali Rahimi, Benjamin Recht, et~al., \emph{Random features for large-scale
  kernel machines.}, NIPS, vol.~3, 2007, p.~5.

\bibitem[RRR16]{RRR16-IOP}
Omer Reingold, Guy~N. Rothblum, and Ron~D. Rothblum, \emph{Constant-round
  interactive proofs for delegating computation}, Proc.\ of the 48th Annual
  {ACM} {SIGACT} Symposium on Theory of Computing, {STOC}, 2016, pp.~49--62.

\bibitem[RV13]{RV13}
Liam Roditty and Virginia {Vassilevska Williams}, \emph{Fast approximation
  algorithms for the diameter and radius of sparse graphs}, Proc. of the 45th
  Annual {ACM} {SIGACT} Symposium on Theory of Computing ({STOC}), 2013,
  pp.~515--524.

\bibitem[SL14]{SL14}
Anshumali Shrivastava and Ping Li, \emph{Asymmetric lsh (alsh) for sublinear
  time maximum inner product search (mips)}, Advances in Neural Information
  Processing Systems, 2014, pp.~2321--2329.

\bibitem[SL15]{SL15}
\bysame, \emph{Asymmetric minwise hashing for indexing binary inner products
  and set containment}, Proc.\ of the 24th International Conference on World
  Wide Web, ACM, 2015, pp.~981--991.

\bibitem[Sol16]{solomon2016fully}
Shay Solomon, \emph{Fully dynamic maximal matching in constant update time},
  Foundations of Computer Science (FOCS), 2016 IEEE 57th Annual Symposium on,
  IEEE, 2016, pp.~325--334.

\bibitem[SU04]{SU04}
S~Cenk Sahinalp and Andrey Utis, \emph{Hardness of string similarity search and
  other indexing problems}, International Colloquium on Automata, Languages,
  and Programming, Springer, 2004, pp.~1080--1098.

\bibitem[TG16]{TG16}
Christina Teflioudi and Rainer Gemulla, \emph{Exact and approximate maximum
  inner product search with lemp}, ACM Transactions on Database Systems (TODS)
  \textbf{42} (2016), no.~1, 5.

\bibitem[Tho68]{Tho68}
Ken Thompson, \emph{Programming techniques: Regular expression search
  algorithm}, Communications of the ACM \textbf{11} (1968), no.~6, 419--422.

\bibitem[Val15]{Valiant15}
Gregory Valiant, \emph{Finding correlations in subquadratic time, with
  applications to learning parities and the closest pair problem}, Journal of
  the ACM (JACM) \textbf{62} (2015), no.~2, 13.

\bibitem[Wil05]{Wil05}
R.~Ryan Williams, \emph{A new algorithm for optimal $2$-constraint satisfaction
  and its implications}, Theoretical Computer Science \textbf{348} (2005),
  no.~2--3, 357--365.

\bibitem[Wil15]{Vass15}
Virginia~Vassilevska Williams, \emph{Hardness of easy problems: basing hardness
  on popular conjectures such as the strong exponential time hypothesis
  (invited talk)}, LIPIcs-Leibniz International Proceedings in Informatics,
  vol.~43, 2015.

\bibitem[Wil16]{Williams16_MA-SETH}
Richard~Ryan Williams, \emph{Strong {ETH} breaks with merlin and arthur: Short
  non-interactive proofs of batch evaluation}, 31st Conference on Computational
  Complexity, {CCC}, 2016, pp.~2:1--2:17.

\bibitem[Wil18]{Wil18}
Ryan Williams, \emph{On the complexity of furthest, closest, and orthogonal
  pairs in low dimensions}, SODA, to appear, 2018.

\bibitem[WMM95]{wu1995subquadratic}
Sun Wu, Udi Manber, and Eugene Myers, \emph{A subquadratic algorithm for
  approximate regular expression matching}, Journal of algorithms \textbf{19}
  (1995), no.~3, 346--360.

\bibitem[WW10]{WW10-subcubic}
Virginia~Vassilevska Williams and Ryan Williams, \emph{Subcubic equivalences
  between path, matrix and triangle problems}, 51th Annual {IEEE} Symposium on
  Foundations of Computer Science, {FOCS} 2010, October 23-26, 2010, Las Vegas,
  Nevada, {USA}, 2010, pp.~645--654.

\end{thebibliography}

\end{document}